\newif\ifreport\reporttrue
\documentclass[journal]{IEEEtran}

\usepackage{setspace}
\usepackage{cite}
\usepackage{graphicx}
\usepackage{caption}
\usepackage{subcaption}
\usepackage[cmex10]{amsmath}
\usepackage{mathtools}
\usepackage{amsthm}
\usepackage{amsfonts}
\usepackage{amssymb}
\usepackage{algorithm}
\usepackage{algorithmic}
\usepackage{bm,xstring,hyperref}
\usepackage{fixltx2e}
\usepackage{tikz}
\usetikzlibrary{decorations.pathreplacing}

\newcommand{\age}{\Delta}

\newcommand{\EE}{\mathbb{E}}

\usepackage{color}

\def\blue{\color{black}}

\def\red{\color{red}}
\newcommand{\ignore}[1]{}
\usepackage[singlelinecheck=false]{caption}

\newtheorem{lemma}{Lemma}

\newtheorem{theorem}{Theorem}
\newtheorem{corollary}{Corollary}
\theoremstyle{definition}

\begin{document}
\title{Optimal Sampling of Brownian Motion for Real-time Monitoring} 
\title{Beyond the Age-of-information: Causal Sampling and  Estimation of Brownian Motion} 

\title{From Age-of-information to MMSE: Real-time Sampling of Brownian Motion} 

\title{Age-of-Information and Tracking of Wiener Process over Channel with Random Delay}

\title{Remote Estimation of the Wiener Process over a Channel with Random Delay}

\title{ Sampling of the Wiener Process for Remote Estimation over a Channel with Random Delay}







\IEEEoverridecommandlockouts
\author{Yin Sun, \emph{Member, IEEE}, Yury Polyanskiy, \emph{Senior Member, IEEE}, and Elif Uysal, \emph{Senior Member, IEEE}

\thanks{Yin Sun was supported in part by NSF grant CCF-18-13050 and ONR grant N00014-17-1-2417. 
Yury Polyanskiy was supported in part by National Science Foundation under Grant No CCF-17-17842, and by the NSF Center for Science of Information (CSoI), under grant agreement CCF-09-39370.
Elif Uysal was supported by TUBITAK through BIDEB fellowship 849044  
and grant number 117E215. This paper was presented in part at IEEE ISIT 2017 \cite{SunISIT2017}.}
\thanks{Y. Sun is with the Department of Electrical and Computer Engineering, Auburn University, Auburn, AL 36849 USA (e-mail: yzs0078@auburn.edu).}
\thanks{Y. Polyanskiy is with the Department of Electrical Engineering and Computer Sciences, Massachusetts Institute of Technology, Cambridge,
MA 02139 USA (e-mail: yp@mit.edu).}
\thanks{E. Uysal is with the Department of Electrical and Electronics Engineering, Middle East Technical
University, Ankara, 06800 Turkey (e-mail: uelif@metu.edu.tr).}
}
\maketitle
\begin{abstract}

In this paper, we consider a problem of sampling a
Wiener process, with samples forwarded to a remote estimator
over  a channel that is modeled as a queue.
The estimator reconstructs an estimate of the  \emph{real-time} signal value from causally received samples.
We study the optimal \emph{online} sampling strategy that
minimizes the mean square estimation error subject to a  sampling rate constraint. 
We prove that the optimal sampling strategy
is a threshold policy, and find the optimal threshold.
This threshold is determined by how much the Wiener process varies during the random service time and the maximum allowed sampling rate. 
Further,
if the sampling times are independent of the observed Wiener process, 
the above sampling problem for minimizing the estimation error is equivalent to a sampling problem for minimizing the age of information. This reveals an interesting connection between the age of information and remote estimation error.
Our comparisons show that the estimation error achieved by the optimal sampling policy can be much smaller than those of age-optimal sampling, zero-wait sampling, and periodic sampling. 

\end{abstract}

\section{Introduction}\label{sec_intro}

In many real-time control  and cyber-physical systems (e.g., airplane/vehicular control, sensor networks, smart grid, stock trading, robotics, etc.), 
 timely updates about the system  status are critical for state estimation and decision making. For example, real-time knowledge about the location, orientation, speed, and acceleration of motor vehicles is imperative for autonomous driving, and  fresh information about stock price, financial news, and interest-rate movements is of paramount importance for stock trading.
In  \cite{Song1990,KaulYatesGruteser-Infocom2012}, the {age of information} was introduced to measure the timeliness of status samples about a remote source. Suppose that the $i$-th status sample is generated at the source at time $S_i$  ($0\leq S_1 \leq S_2\leq\ldots$) and is delivered to the destination at time $D_i$. At time $t$, the freshest sample available at the destination was generated at time $U(t) = \max\{S_i: D_i\leq t\}$. The \emph{age of information}, or simply the \emph{age}, is a function of time $t$ that is defined as
\begin{align}\label{eq_age_definition}
\age(t)=t- U(t) = t -\max\{S_i: D_i\leq t\},
\end{align} 
which is the {time difference} between the generation time $U(t)$ of the freshest received sample and the current time $t$. 
Hence, a small age $\age(t)$ implies that there exists a fresh status sample at the destination. As plotted in Fig. \ref{fig:age1}, the age increases linearly over time and is reset to a smaller value once a new sample is received. Hence, the age  $\age(t)$ exhibits a sawtooth pattern. Recently, the age of information concept has received significant attention, because of the rapid growth of real-time
applications. A number of status update policies have been developed to keep the age $\age(t)$ small, subject to constraints on limited network resources, e.g., \cite{2012ISIT-YatesKaul,Yates2016,KaulYatesGruteser-Infocom2012,2015ISITYates,Kam-status-ISIT2013,Bacinoglu2015,SunInfocom2016,report_AgeOfInfo2016,Sun2019,Costa2016,Kam2016,Bedewy2016,Bedewy2017,Kadota2016,Bacinoglu2017}.




\begin{figure}
\centering
\begin{tikzpicture}[scale=0.21]
\draw [<-|] (0,11)  -- (0,0) -- (14.5,0);
\draw [|->] (15,0) -- (30.5,0) node [below] {\small$t$};
\draw (-2,12) node [right] {\small$\age(t)$};
\draw
(0,0) node [below] {\small$S_0$}
(8,0) node [below] {\small$S_1$}
(17,0) node [below] {\small$S_{j-1}$}
(24,0) node [below] {\small$S_{j}$};
\fill
(8,0)  circle[radius=4pt]
(17,0)  circle[radius=4pt]
(24,0)  circle[radius=4pt]
(0,0)  circle[radius=4pt]
(4,0)  circle[radius=4pt]
(11,0)  circle[radius=4pt]
(20,0)  circle[radius=4pt]
(27,0)  circle[radius=4pt];
\draw
(4,0) node [below] {\small$D_0$}
(11,0) node [below] {\small$D_1$}
(21,0) node [below] {\small$D_{j-1}$}
(27,0) node [below] {\small$D_{j}$};
\draw[ thick, domain=0:4] plot (\x, {\x+3})  -- (4, {4});
 \draw [ thick, domain=4:11] plot (\x, {\x})  -- (11, {3});
\draw[ thick, domain=11:14] plot (\x, {\x-8});
\draw[ thick, domain=16:20] plot (\x, {\x-14}) -- (20, {3});
\draw[ thick, domain=20:27] plot (\x, {\x-17}) -- (27, {3});
\draw[ thick, domain=27:29] plot (\x, {\x-24});
\draw[  thin,dashed,  domain=0:4] plot (\x, {\x})-- (4, 0);
\draw[  thin,dashed,  domain=8:11] plot (\x, {\x-8})-- (11, 0);
\draw[  thin,dashed,  domain=17:20] plot (\x, {\x-17})-- (20, 0);
\draw[  thin,dashed,  domain=24:27] plot (\x, {\x-24})-- (27, 0);
\end{tikzpicture}
\caption{Evolution of the age of information $\age(t)$ over time.}
\label{fig:age1}
\end{figure}


In practice, the state of  many systems is in the form of a time-varying signal $W_t$, such as the location of a vehicle, the wind speed of a hurricane, and the price chart of a stock. 
These signals may change slowly at some time
 and vary more dynamically later.
Hence, the time difference between the source and destination, described by the age $\age(t)=t- U(t)$, cannot fully characterize the amount of change $W_t - W_{U(t)}$ in the signal value. This motivated us to go beyond the age of information concept and  investigate \emph{timely updates of signal samples}. 

%





%
%
%
%
%
%
%
%

Let us consider a status update system with two terminals (see Fig.~\ref{fig_model}): 
An observer taking samples from a continuous-time signal $W_t$ which is modeled as a Wiener process, and an estimator, whose goal is to provide the best-guess $\hat W_t$ for the real-time signal value $W_t$ at all time $t$.\footnote{This paper focuses on a Wiener process signal model, which has some nice properties that were used our analysis. 
An important future direction is to study more general signal models. A recent result along this direction was reported  in \cite{Ornee2019}.}
The two terminals are connected by a channel 
that transmits time-stamped samples of the form $(S_i, W_{S_i})$ according to a first-in, first-out (FIFO) order, where $S_i$ is the sampling time of the $i$-th sample and $W_{S_i}$ is the value of the $i$-th sample. The samples are stored in a queue while they wait to be served by the channel. 
We assume that the samples experience \emph{i.i.d.} random transmission times over the channel, which may be caused by fading, interference, collisions, retransmissions, and etc. 
As such, the channel is modeled as a FIFO queue with \emph{i.i.d.} service time $Y_i$ satisfying $\mathbb{E}[Y_i^2]<\infty$, where $Y_i\geq 0$ is the transmission time of sample $i$. This queueing model is helpful to understand the robustness of remote estimation and control systems under occasionally  slow service. 
For example, a UAV flying by a WiFi access point may run into a communication outage caused by interference from the access point. The resulting delay in packet reception may affect the stability of UAV flight control and navigation \cite{Azari2019}.

Let $G_i$ be the service starting time of sample $i$ such that $S_i \leq G_i$. 
The delivery time of sample $i$ is $D_i = G_i + Y_i$. The initial value $W_0= 0$ is known by the estimator for free, which is represented by $S_0 =D_0 = 0$. At any time $t$, the estimator forms an estimate $\hat W_t$ using the samples received up to time $t$. Similar with \cite{SOLEYMANI20161}, we assume that the estimator neglects the implied knowledge when no sample was delivered. 
The quality of remote estimation is evaluated via the time-average mean-square error (MSE) between $W_t$ and $\hat W_t$:
\begin{align}\label{eq_MSE11111}
\mathsf{mse} = \limsup_{T\to \infty} {1\over T} \EE \left[\int_0^T (W_t - \hat W_t)^2 dt\right]. 
\end{align}
The sampler is subject to a sampling rate constraint
\begin{align} \label{eq_rate_constraint}
\liminf_{n\to\infty} {1\over n} \mathbb{E}[S_n] \ge {1\over f_{\max}}, 
\end{align}
where $f_{\max}$ is the maximum allowed sampling rate.
In practice, the sampling rate constraint \eqref{eq_rate_constraint} is imposed when there is a need to reduce the cost (e.g., energy consumption) for the transmission, storage, and processing of the samples. 

\begin{figure}
\centering
\includegraphics[width=0.47\textwidth]{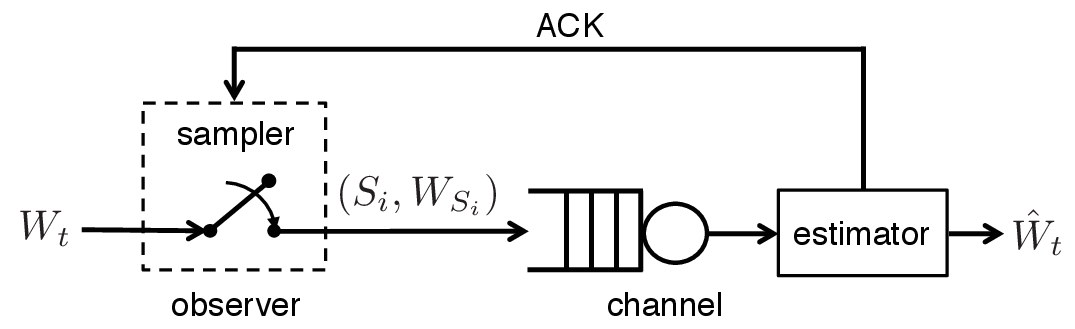}   
\caption{System model.}\vspace{-0.0cm}
\label{fig_model}
\end{figure}    
Our goal is to find an optimal \emph{online} sampling strategy that minimizes the MSE in \eqref{eq_MSE11111} by choosing the sampling times $S_i$ \emph{causally} subject to the sampling rate constraint \eqref{eq_rate_constraint}. 
The  contributions of this paper are summarized as follows:
\ifreport
\begin{itemize}
\else
\begin{itemize}[leftmargin=*]
\fi
\item We formulate the optimal sampling problem as a constrained continuous-time Markov decision problem with a continuous state space, and solve it exactly. 
We prove that the optimal online sampling strategy for the Wiener process is a threshold policy\footnote{
A sampling policy is said to be a \emph{threshold} policy if a new sample is taken when a threshold condition is satisfied. Examples of threshold policies can be found in Section \ref{sec:samplingpolicies}.}, and find the optimal threshold.
Let $Y$ be a random variable with the same distribution as $Y_i$. The optimal threshold is determined by $f_{\max}$ and $W_Y$, where $W_Y$ is a random variable that has the same distribution as the amount of signal variation $(W_{t+Y} - W_t)$ that occurs during the random service time $Y$. The random variable $W_Y$  indicates a tight \emph{coupling},   in the optimal sampling policy, between the source process $W_t$ and the service time $Y$. 


\item Our threshold-based optimal sampling policy has an important difference from the previous threshold-based sampling policies studied in, e.g., \cite{Astrom2002,Hajek2002,Hajek2008,Lipsa2009, Nuno2011,Molin2012, nayyar2013,Rabi2012,Basar2014,GaoACC2015, GaoCDC2015, Gao2016,GaoACC2016,GaoAsilomar2016,GaoCDC2016,CHAKRAVORTY2016,ChakravortyISIT2017,Chakravorty2018}: {\blue 
We have proven that it is better to not take any new sample when the server is busy. Consequently, the threshold should be \emph{disabled} when the server is busy and \emph{reactivated} once the server becomes available again.} This is one of the reasons that sampling policies that ignore the state of the server, such as periodic sampling, can have a large estimation error.



%



\item 
 We show, perhaps surprisingly, even in the absence of a sampling rate constraint
(i.e., $f_{\max}=\infty$), the optimal sampling strategy is \emph{not} zero-wait sampling in which a new sample is generated once the previous sample is delivered; {\blue rather, it is optimal to wait for a certain amount of time after the previous sample is delivered, and then take the next sample.}

\item Our study reveals a relationship between the age of information and the estimation error of Wiener process: If the sampling times $S_i$ are independent of the observed Wiener process (i.e., the sampling times $S_i$ are chosen without using any information about the Wiener process), the MSE in \eqref{eq_MSE11111} is exactly equal to the time-average expectation of the age of information $\limsup_{T\to \infty} {1\over T} \EE [\int_0^T \Delta(t) dt]$. Hence, the sampling problem for minimizing the MSE is equivalent to a sampling problem for minimizing the age, where the second problem was solved recently in \cite{SunInfocom2016,report_AgeOfInfo2016,Sun2019}. If the sampling times $S_i$ are chosen based on causal knowledge of the Wiener process, the age-optimal sampling policy (i.e., the sampling policy that minimizes the time-average expected age of information) no longer minimizes the MSE: Specifically, in the age-optimal sampling policy, a new sample is taken only when the \emph{age of information} $\age(t)$, or equivalently the \emph{expected estimation error} $\mathbb{E}[(W_t - \hat W_t)^2]$, is no smaller than a threshold; while in the MSE-optimal sampling policy, a new sample is taken only when \emph{the instantaneous estimation error} $|W_t - \hat W_t|$ is no smaller than  a threshold. The asymptotics of the MSE-optimal and  age-optimal sampling policies at  long/short service time or low/high sampling rates are also studied. 

%

%

\item Our theoretical and numerical comparisons show that the MSE of the optimal sampling policy can be much smaller than those of age-optimal sampling, periodic sampling, and the zero-wait sampling policy described in \eqref{eq_Zero_wait} below. In particular, periodic sampling is far from optimal when the sampling rate is sufficiently low or sufficiently high; age-optimal  sampling is far from optimal when the sampling rate is sufficiently low; periodic sampling, age-optimal sampling, and zero-wait sampling policies are all far from optimal if {the service time distribution is heavy-tailed}.

\end{itemize}


The rest of this paper is organized as follows. In Section \ref{sec:relate}, we discuss some related work. In Section \ref{sec:model}, we describe
the system model and the formulation of the optimal sampling problem. In Section \ref{sec:solution}, {we present the solution to this problem} and compare it with some other sampling policies. In Section \ref{sec_proof}, we describe the proof of this optimal solution. Some simulation results are provided in Section \ref{sec:simulation}. 

\section{Related Work}\label{sec:relate}

Lossy source coding and the rate-distortion function of the Wiener process was studied in, e.g., \cite{Berger1970,Kipnis2016}, where the rate-distortion function represents the optimal tradeoff between the source coding rate and the distortion  (i.e., MSE) for recovering of the Wiener process. The goal of these studies is to reconstruct the realization of Wiener process during a past time interval with a small distortion, which can be regarded  as an \emph{offline} signal reconstruction problem. 
This differs from our \emph{online} signal tracking problem, where the real-time value of the Wiener process is estimated at the destination from causally received samples. 

This paper is related to recent studies on the age of information, e.g., \cite{2012ISIT-YatesKaul,Yates2016,KaulYatesGruteser-Infocom2012,2015ISITYates,Kam-status-ISIT2013,Bacinoglu2015,SunInfocom2016,report_AgeOfInfo2016,Sun2019,Costa2016,Kam2016,Bedewy2016,Bedewy2017,Kadota2016,Bacinoglu2017}. As mentioned above in Section \ref{sec_intro}, a connection between the age of information and the estimation error of Wiener process is characterized in this paper. The estimation error of Wiener process was also mentioned in \cite{2012ISIT-YatesKaul,Yates2016} as an illustration of the age of information, where age-based sampling was not studied and the condition that the sampling times are independent of the Wiener process was used implicitly. 
Recently, a relationship between a nonlinear function of the age of information and the estimation error of the Ornstein-Uhlenbeck (OU) process was found in a follow-up study of the current paper \cite{Ornee2019}.

This paper can also be considered as a contribution to 
the rich literature on remote estimation, e.g., \cite{Astrom2002,Hajek2002,Hajek2008,Lipsa2009, Nuno2011,Molin2012, nayyar2013,Rabi2012,Basar2014,GaoACC2015, GaoCDC2015, Gao2016,GaoACC2016,GaoAsilomar2016,GaoCDC2016,CHAKRAVORTY2016,ChakravortyISIT2017,Chakravorty2018}, by including a queueing model. In \cite{Astrom2002}, \r{A}str\"{o}m and Bernhardsson showed that a threshold-based sampling method, in which a new sample is taken once the amount of signal variation since the previous sample has reached a threshold, can achieve a smaller estimation error than the traditional periodic sampling method with the same sampling rate. Such a threshold-based sampler and a Kalman-like estimator have been proven to be jointly optimal for minimizing the remote estimation error of several discrete-time signal processes in \cite{Hajek2002,Hajek2008,Lipsa2009,Nuno2011,Molin2012,nayyar2013}. The sampling and remote estimation of continuous-time signal processes were considered in \cite{Rabi2012,Basar2014}, where it was shown that a threshold-based sampling policy is optimal for minimizing the estimation error of the Wiener process, and the optimal threshold was found. In \cite{Hajek2002,Hajek2008,Lipsa2009,Nuno2011,Molin2012,nayyar2013,Rabi2012,Basar2014}, it was assumed that the samples are transmitted from the sampler to the estimator over a perfect channel that is error and noise free. 
There are some recent studies that used explicit channel models. In \cite{GaoACC2015, GaoCDC2015, Gao2016}, Gao et. al. considered the optimal transmission scheduling and remote estimation of an \emph{i.i.d.} discrete-time source process $X_t$ over an additive noise channel. Because of the noise, the transmitter needs to encode its message before transmission. In \cite{GaoACC2015, GaoCDC2015}, it was shown that, for a class of symmetric probability distributions on the source symbol $X_t$, if the transmission scheduling policy is threshold-based, i.e., a new coded packet is sent if $|X_t|$ is no smaller than a threshold, then the optimal encoder and decoder are piecewise affine. In \cite{Gao2016}, it was shown that if (i) the encoder and decoder (i.e., estimator) are piecewise affine, and (ii) the transmission scheduler satisfies some technical assumption, the optimal transmission scheduling policy is threshold-based. Some extensions of this research were reported in  \cite{GaoACC2016,GaoAsilomar2016,GaoCDC2016}.
In \cite{CHAKRAVORTY2016,ChakravortyISIT2017,Chakravorty2018}, Chakravorty and Mahajan considered optimal transmission scheduling and remote estimation over a few channel models, where it was proved that a threshold-based transmission policy and a Kalman-like estimator are jointly optimal for minimizing the remote estimation error. 

The closest study to this paper are \cite{Rabi2012,Basar2014}, where the optimal sampler of the Wiener process was designed in the absence of queueing and random service time (i.e., $Y_i = 0$). As we will see later, the queueing model affects the structure of the optimal sampler. Specifically, the sampler should disable the threshold when there is a packet in service and {reactivate} the threshold after all previous packets are delivered. A novel proof procedure is developed in the current paper to find the optimal sampler design.

\section{System Model and Problem Formulation}\label{sec:model}


\subsection{MMSE Estimation Policy}
At time $t$, the information available to the estimator contains two part: (i) $M_t = \{(S_i, W_{S_i}, D_i): D_i \leq t\}$, which contains the sampling time $S_i$, sample value $W_{S_i}$, and delivery time $D_i$ of the samples delivered by time $t$ and (ii) the facts that no sample has been received after the latest sample delivery time $\max\{D_i: D_i\leq t\}$. Similar with \cite{SOLEYMANI20161}, we assume that the estimator neglects the implied knowledge when no sample was delivered. In this case, the minimum mean-square error (MMSE) estimation policy  \cite{Poor:1994} is given by 
(see Appendix \ref{app_estimation} for its derivation)  
\begin{align}\label{eq_esti}
\hat{W}_{t}  = & \mathbb{E}[{W}_t | M_t ] \nonumber\\
                = & W_{S_i}, ~\text{if}~t\in[D_i,D_{i+1}),~i=0,1,2,\ldots, 
\end{align}
which is illustrated in Fig. \ref{fig_signal}(b). 

\begin{figure}[t!]
    \centering
    \begin{subfigure}[t]{0.5\textwidth}
        \centering
        \includegraphics[width=0.6\textwidth]{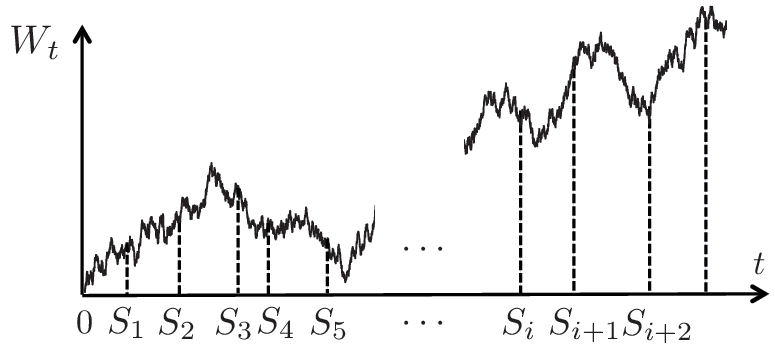}
        \caption{Wiener process $W_{t}$ and its samples.}\label{fig_signal1}
    \end{subfigure}\\
    \begin{subfigure}[t]{0.5\textwidth}
        \centering
        \includegraphics[width=0.6\textwidth]{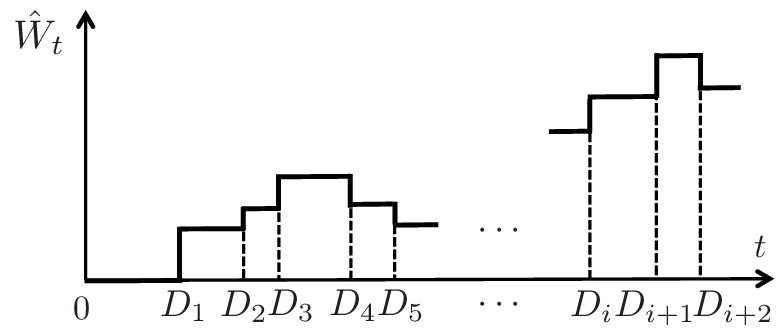}
        \caption{Estimate process $\hat{W}_{t}$ using causally received samples.}\label{fig_signal2}
    \end{subfigure}
                
    \caption{Illustration of the MMSE estimation policy \eqref{eq_esti}.}
    \ifreport
    \else
\fi\label{fig_signal}
\end{figure}











\subsection{Sampling Policies}\label{sec:samplingpolicies}
Let $I_t \in\{0,1\}$ denote the idle/busy state of the server at time $t$. As shown in Fig. \ref{fig_model}, the server state $I_t$ is known by the sampler through acknowledgements (ACKs). We assume that once a sample is delivered to the estimator, an ACK is fed back to the sampler with zero delay. Hence, the information that is available to the sampler at time $t$ can be expressed as $\{W_s, I_s: 0\leq s\leq t\}$.

In online sampling policies, each sampling time $S_i$ is chosen {causally} using the information available at the sampler. To characterize this statement precisely, 
we define 
\begin{align}
\mathcal{N}_t &= \sigma(W_s, I_s: 0\leq s\leq t),~\mathcal{N}_t^+ = \cap_{s>t}\mathcal{N}_s, 
\end{align}
where $\sigma(X_1, X_2,\ldots, X_n)$ represents the $\sigma$-field generated by the random variables $X_1, X_2,\ldots, X_n$.  
Then, $\{\mathcal{N}_t^+,t\geq0\}$ is a \emph{filtration} (i.e., a {non-decreasing} and {right-continuous} family of $\sigma$-fields) of the information available at the sampler.
Each sampling time $S_{i} $ is a \emph{stopping time} with respect to the filtration $\{\mathcal{N}_t^+,t\geq0\}$, i.e., 
\begin{align}\label{eq_sampling_policy}
\{S_{i}\leq t\} \in \mathcal{N}_t^+,~\forall t\geq0.
\end{align}

Let $\pi=(S_1,S_2,\ldots)$ denote a sampling policy where $S_1\leq S_2\leq \cdots$ form an increasing sequence of sampling times. Let $\Pi$ denote a set of \emph{online} (also called \emph{causal}) sampling policies  
 satisfying the following two conditions: \emph{(i)} Each sampling policy $\pi\in\Pi$ satisfies \eqref{eq_sampling_policy}  for all $i=0,1,\ldots$  
\emph{(ii)} The inter-sampling times $\{T_i = S_{i+1}-S_i, i=0,1,\ldots\}$ form a \emph{regenerative process} \cite[Section 6.1]{Haas2002}: There exist an increasing sequence $0\leq {k_1}<k_2< \ldots$ of almost surely finite random integers such that the post-${k_j}$ process $\{T_{k_j+i}, i=0,1,\ldots\}$ has the same distribution as the post-${k_1}$ process $\{T_{k_1+i}, i=0,1,\ldots\}$ and is independent of the pre-$k_j$ process $\{T_{i}, i=0,1,\ldots, k_j-1\}$; in addition, $\mathbb{E}[{k_{j+1}}-{k_j}]<\infty$, $\mathbb{E}[S_{k_{1}}^2]<\infty$, and $0<\mathbb{E}[(S_{k_{j+1}}-S_{k_j})^2]<\infty$  for  $j=1,2,\ldots$ 
By Condition (ii), we can obtain that, almost surely, 
\begin{align}\label{eq_infinite}
\lim_{i\rightarrow\infty} S_i=\infty,~\lim_{i\rightarrow\infty} D_i = \infty.
\end{align}
We analyze the MSE in \eqref{eq_MSE11111}, but operationally a nicer criterion is $\limsup_{n\rightarrow \infty} $ ${\mathbb{E}[\int_0^{D_n} (W_{t}-\hat{W}_{t})^2dt]}/{\mathbb{E}[D_n]}$. These two criteria are associated to two definitions of ``average cost per unit time'' used in the literature of infinite-horizon undiscounted semi-Markov decision problems  \cite{Ross1970,Mine1970,Hayman1984,Arapostathis1993Survey,Bertsekas2005bookDPVol1}. They are equivalent, if $\{T_1, T_2, \dots\}$ is a regenerative process, or more generally, if $\{T_1, T_2, \dots\}$ has only one ergodic class \cite{Ross1970,Mine1970,Hayman1984}. If no condition is imposed, however, these two criteria are different.

Some examples of the sampling policies in $\Pi$ are:
\begin{itemize}
\vspace{-0.5ex}
\ifreport 
\item[1.] \emph{Periodic sampling} \cite{Nyquist1928,Shannon1949}: 
\fi
The inter-sampling times are constant, such that for some $\beta\geq 0$,
\begin{align} \label{eq_uniform}
S_{i+1} = S_i+ \beta.
\end{align}

\item[2.] \emph{Zero-wait sampling} \cite{2015ISITYates,SunInfocom2016,report_AgeOfInfo2016,KaulYatesGruteser-Infocom2012}: A new sample is generated once the previous sample is delivered, i.e.,
\begin{align} \label{eq_Zero_wait}
S_{i+1} = S_i+ Y_i.
\end{align}

\item[3.] \emph{Threshold policy on expected estimation error} \cite{2015ISITYates,SunInfocom2016,report_AgeOfInfo2016}: The sampling times are given by
\begin{align}\label{eq_thm2_11}
S_{i+1}=&\inf \left\{ t\geq S_i + Y_i: t - {S_i} \geq{\beta}\right\} \\
=&\inf \left\{ t\geq S_i + Y_i: \mathbb{E}[(W_t - \hat W_t)^2] \geq{\beta}\right\},\label{eq_thm2_1}
\end{align}
where $S_i + Y_i = D_i$ and, according to \eqref{eq_esti}, $\hat W_t = W_{S_i}$. 
\item[4.] \emph{Threshold policy on instantaneous estimation error}: 
The sampling times are given by
\begin{align}\label{eq_opt_solution}
S_{i+1}= \inf \left\{ t\geq S_i + Y_i: |W_t - \hat W_{t}| \!\geq\! \sqrt{\beta}\right\},
\end{align} 
where $\hat W_t = W_{S_i}$.
The sampling policy in \eqref{eq_opt_solution} can be understood as follows:
As illustrated in Fig. \ref{fig_policy1}, if $|W_{S_i+Y_i} - W_{S_i}| \geq \sqrt{\beta}$, 
sample $i+1$ is generated at the time $S_{i+1} = S_i+Y_i$ when sample $i$ is delivered; otherwise, if $|W_{S_i+Y_i} - W_{S_i}| < \sqrt{\beta}$, sample $i+1$ is generated at the earliest time $t$ such that $t\geq S_i+Y_i$ and 
$|W_t - W_{S_i}|$ reaches the threshold $\sqrt{\beta}$. {\blue It is worthwhile to emphasize that even if there exists time $t\in[S_i, S_i+Y_i)$ such that $|W_{t} - W_{S_i}| \geq \sqrt{\beta}$, no sample is taken at such time $t$, as depicted in both cases of Fig. \ref{fig_policy1}. 
In other words, the threshold-based control is \emph{disabled} during $[S_i, S_i+Y_i)$ and is \emph{reactivated} at time $S_i+Y_i$.} 
\end{itemize}

\begin{figure}[t!]
\centering
    \!\!\!\!\begin{subfigure}[t]{0.23\textwidth}
        \centering
        \includegraphics[width=0.9\textwidth]{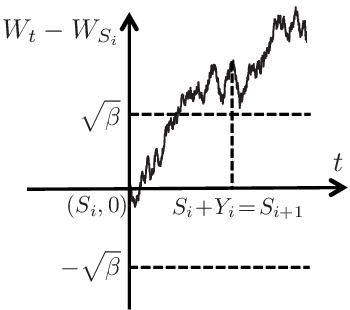}
        \caption[(i)]{If $|W_{S_i+Y_i}\! -\! W_{S_i}|\! \geq\! \sqrt{\beta}$,  sample $i+1$ is taken at  time $S_{i+1} = S_i+Y_i$.}\label{}
    \end{subfigure}~~~
    \begin{subfigure}[t]{0.23\textwidth}
        \centering
        \includegraphics[width=0.9\textwidth]{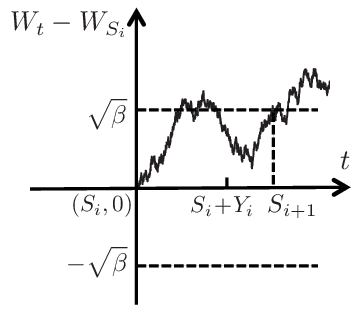}
        \caption{If $|W_{S_i+Y_i} \!-\! W_{S_i}|\! <\! \sqrt{\beta}$, sample $i+1$ is taken at time $t$ that satisfies $t\geq S_i+Y_i$ and $|W_{t}\! -\! W_{S_i}|=\sqrt{\beta}$.}\label{}
    \end{subfigure}
  
    \caption{Illustration of the threshold-based sampling policy \eqref{eq_opt_solution}, where no sample is taken during $[S_i,S_i+Y_i)$. }\label{fig_policy1}
\end{figure}

A sampling policy $\pi\in\Pi$ is said to be \emph{signal-ignorant} (\emph{signal-aware}), if $\pi$ is (not) independent of the Wiener process $\{W_t,t\geq0\}$. The sampling policies  \eqref{eq_uniform}, \eqref{eq_Zero_wait}, and \eqref{eq_thm2_1} are signal-ignorant, and the sampling policy  \eqref{eq_opt_solution} is signal-aware. 


\subsection{Optimal  Sampling  Problem}
We assume that the source process $\{W_t, t\geq 0\}$ and the service times $\{Y_i,i = 1,2,\ldots\}$  are \emph{mutually independent} and  
do not change according to the sampling policy. In addition, we assume that the $Y_i$'s are \emph{i.i.d.} with $\mathbb{E}[Y_i^2]<\infty$. 
The optimal sampling problem for minimizing the MSE  subject to a sampling rate constraint is formulated as 
\begin{align}\label{eq_DPExpected}
\mathsf{mse}_{\text{opt}}\triangleq\inf_{\pi\in\Pi}~& \limsup_{T\rightarrow \infty}\frac{1}{T}\mathbb{E}\left[\int_0^{T} (W_t - \hat W_t)^2dt\right] \\
~\text{s.t.}~~& \liminf_{n\rightarrow \infty} \frac{1}{n} 
\mathbb{E}[S_n]\geq \frac{1}{f_{\max}},\label{eq_constraint}
\end{align}
where $\mathsf{mse}_{\text{opt}}$ denotes the optimal  value of \eqref{eq_DPExpected}. Later on in the paper, the unconstrained problem with $f_{\max} = \infty$ will also be studied.

\section{Optimal Sampling Policies}\label{sec:solution}
\subsection{Signal-aware Sampling}

Problem  \eqref{eq_DPExpected} is a constrained continuous-time Markov decision problem with a continuous state space. Such problems are often lack of closed-form or analytical solutions, however
we were able to  solve \eqref{eq_DPExpected} exactly:
\begin{theorem}\label{thm_1}
If the service times $Y_i$'s are {i.i.d.} with $\mathbb{E}[Y_i^2]<\infty$, then there exists $\beta\geq0$ such that the sampling policy \eqref{eq_opt_solution} is an optimal solution of
  \eqref{eq_DPExpected},
 and the optimal $\beta$ is determined by solving\footnote{If $\beta\rightarrow0$, the last terms in \eqref{eq_thm1} and \eqref{eq_thm2} are determined by L'Hospital's rule.}
\begin{align} \label{eq_thm1}
\mathbb{E}[\max(\beta,W_Y^2)] \!=\! \max\left(\frac{1}{f_{\max}}, \frac{\mathbb{E}[\max(\beta^2,W_Y^4)]}{2\beta}\right),\!
\end{align}
where $Y$ is a random variable with the same distribution as $Y_i$.
The optimal  value of \eqref{eq_DPExpected} is then given by \emph{
\begin{align}\label{thm_1_obj}
\mathsf{mse}_{\text{opt}}=\frac{\mathbb{E}[\max(\beta^2,W_Y^4)]}{6\mathbb{E}[\max(\beta,W_Y^2)]} + \mathbb{E}[Y].
\end{align}}
\end{theorem}

\begin{proof}
See Section \ref{sec_proof}.
\end{proof}



According to Theorem \ref{thm_1}, in the optimal signal-aware sampling policy, the $(i+1)$-th sample is taken at the earliest time $t$ satisfying two conditions:  \emph{(i)} The $i$-th sample has already been delivered by time $t$, i.e., $t\geq D_i = S_i + Y_i$, and \emph{(ii)} 
the instantaneous estimation error $|W_t - \hat W_t|$ at time $t$ is no smaller than a threshold $\sqrt{\beta}$. In addition, the threshold $\sqrt{\beta}$ is determined by the maximum allowed sampling rate $f_{\max}$ and $W_Y$, where $W_Y$ is a random variable that has the same distribution with the amount of signal variation $(W_{t+Y} - W_t)$ during the random service time $Y$ for all starting time $t$. This indicates a tight \emph{coupling} between the source process $W_t$ and the service time $Y$,  in the optimal sampling policy.  

Equation \eqref{eq_thm1} can be solved by using the bisection method with a low computational complexity. Hence, Problem  \eqref{eq_DPExpected} does not suffer from 
the curse of dimensionality 
encountered in most Markov decision problems with continuous state spaces. We note that the  sampling policy in \eqref{eq_opt_solution} and \eqref{eq_thm1} is quite general in the sense that it is optimal
 for any service time distribution satisfying $\mathbb{E}[Y^2]<\infty$. The optimal signal-aware sampling policy in \eqref{eq_opt_solution} and \eqref{eq_thm1} is also called the ``MSE-optimal'' sampling policy in the sequel.




\subsection{Signal-ignorant Sampling and the Age of Information}\label{sec_signal_independent}

Let $\Pi_{\text{signal-ignorant}}\subset \Pi$ denote the set of signal-ignorant sampling policies, defined as   
\begin{align}
\Pi_{\text{signal-ignorant}} \!=\! \{\pi\!\in\Pi: \pi \text{ is independent of }\{W_t, t\geq 0\}\}.
\end{align}
In these policies, the sampling decisions depend only  on the service time $\{Y_i,i=1,2,\ldots\}$ but not the source process $\{W_t, t\geq 0\}$.
For each $\pi \in \Pi_{\text{signal-ignorant}}$,  the objective function in \eqref{eq_DPExpected} can be rewritten as 
\ifreport
(see Appendix \ref{app_age} for the proof)
\fi
{
\begin{align}\label{eq_age_MSE}
& \limsup_{T\rightarrow \infty}\frac{1}{T}\mathbb{E}\left[\int_0^{T} (W_t - \hat W_t)^2dt\right] \nonumber\\
= &\limsup_{T\rightarrow \infty}\frac{1}{T}\mathbb{E}\left[ \int_0^{T}\Delta (t) dt\right],
\end{align}}
\!\!\!where $\Delta (t)$ is the \emph{age of information} defined in \eqref{eq_age_definition}. In FIFO queueing systems, $D_i \leq D_{i+1}$ holds for all $i$. Hence, the age $\Delta(t)$ can be equivalently expressed as
\begin{align}\label{eq_age_def}
\Delta (t) = t - S_i, ~t\in[D_i,D_{i+1}),~i=0,1,2,\ldots
\end{align}

If the set of feasible policies is restricted from $\Pi$ to $\Pi_{\text{signal-ignorant}}$, \eqref{eq_DPExpected} reduces to the following sampling problem for minimizing the time-average expectation of the age of information \cite{SunInfocom2016,report_AgeOfInfo2016,Sun2019}:
\begin{align}\label{eq_age}
\mathsf{mse}_{\text{age-opt}}\triangleq\inf_{\pi \in\Pi_{\text{signal-ignorant}} }& \limsup_{T\rightarrow \infty}\frac{1}{T}\mathbb{E}\left[ \int_0^{T}\Delta (t) dt\right] \\
~\text{s.t.}~~~~~ &\liminf_{n\rightarrow \infty} \frac{1}{n} 
\mathbb{E}[S_n]\geq \frac{1}{f_{\max}},\nonumber
\end{align}
where $\mathsf{mse}_{\text{age-opt}}$ denotes the optimal  value of \eqref{eq_age}. Because $\Pi_{\text{signal-ignorant}}\subset \Pi$,
\begin{align}\label{eq_compare_aware}
\mathsf{mse}_{\text{opt}} \leq \mathsf{mse}_{\text{age-opt}}.
\end{align}
Note that problem \eqref{eq_age} is simpler than \eqref{eq_DPExpected} because the sampler does not use knowledge of $W_t$ to make decisions. To solve \eqref{eq_DPExpected}, stronger techniques than those in \cite{SunInfocom2016,report_AgeOfInfo2016,Sun2019} are developed in Section \ref{sec_proof}.




%
\begin{theorem}\cite{Sun2019}\label{thm_2}
If the service times $Y_i$'s are {i.i.d.} with $\mathbb{E}[Y_i^2]<\infty$, then there exists $\beta\geq0$ such that 
the sampling policy \eqref{eq_thm2_1} is an optimal solution of \eqref{eq_age}, and the optimal $\beta$ is determined by solving
\begin{align}\label{eq_thm2}
\mathbb{E}[\max(\beta,Y)] \!=\! \max\left(\frac{1}{f_{\max}}, \frac{\mathbb{E}[\max(\beta^2,Y^2)]}{2\beta}\right),\!
\end{align}
where $Y$ is a random variable with the same distribution as $Y_i$.
The optimal  value of \eqref{eq_age} is then given by\emph{
\begin{align}\label{thm_2_obj}
\mathsf{mse}_{\text{age-opt}}\triangleq\frac{\mathbb{E}[\max(\beta^2,Y^2)]}{2\mathbb{E}[\max(\beta,Y)] } + \mathbb{E}[Y].
\end{align}}
\end{theorem}

Theorem \ref{thm_2} was proven in \cite{SunInfocom2016,report_AgeOfInfo2016} under an extra condition that the time difference
$S_{i+1}-D_i$ is upper bounded by a constant $M>0$. In \cite{Sun2019},
Theorem \ref{thm_2} was established without requiring this extra condition.

One can obtain some interesting observations by comparing Theorem \ref{thm_1} and Theorem \ref{thm_2}: 
In the optimal signal-ignorant sampling policy presented in Theorem \ref{thm_2}, the $(i+1)$-th sample is taken at the earliest time $t$ satisfying two conditions: \emph{(i)} The $i$-th sample has already been delivered by time $t$, i.e., $t\geq D_i = S_i + Y_i$, and \emph{(ii)} the expected estimation error $\mathbb{E}[(W_t - \hat W_t)]^2$ at time $t$, which, by \eqref{eq_thm2_11} and \eqref{eq_age_def}, is equal to the age $\Delta(t)$, is no smaller than a threshold $\beta$. The first condition is the same with that in Theorem \ref{thm_1}, but the second condition is quite different: Because the sampler has no knowledge about the Wiener process (except for its distribution), it can only use expected estimation error to make decisions. Further, the threshold ${\beta}$ in Theorem \ref{thm_2} is determined by the maximum allowed sampling rate $f_{\max}$ and the random service time $Y$, which is also different from the case in Theorem \ref{thm_1}.
The optimal signal-ignorant sampling policy in \eqref{eq_thm2_1} and \eqref{eq_thm2} is also referred to as the ``age-optimal'' sampling policy. 

%

In the following, the asymptotics of the MSE-optimal and age-optimal sampling policies at  low/high service time or low/high sampling frequencies are studied. 
\subsection{Short Service Time or Low Sampling Rate}
\ifreport
Let 
\begin{align}\label{eq_scale_Y_i}
Y_i = \alpha X_i
\end{align}
\fi
represent the scaling of the service time $Y_i$ with $\alpha$, where $\alpha\geq0$ and the $X_i$'s are \emph{i.i.d.} positive random variables. If $\alpha\rightarrow  0$ or $f_{\max}\rightarrow 0$, we can obtain from \eqref{eq_thm1} that
\ifreport
(see Appendix \ref{app_low_delay} for the proof)
\fi
\begin{align}\label{eq_betaasy}
\beta = \frac{1}{f_{\max}} + o\left(\frac{1}{f_{\max}}\right),
\end{align}
where $f(x)=o(g(x))$ as $x\rightarrow a$ means that $\lim_{x\rightarrow a} $ $f(x)/g(x) = 0$.
In this case, the MSE-optimal sampling policy 
\ifreport
in \eqref{eq_opt_solution} and \eqref{eq_thm1} 
\fi
becomes  
\begin{align} \label{eq_opt_solution1}
S_{i+1}\! =\! \inf \left\{ t \geq S_i: |W_t - W_{S_i}| \!\geq\! \sqrt{\frac{1}{f_{\max}}}\right\},
\end{align}
and 
\ifreport 
as shown in Appendix \ref{app_low_delay}, the optimal value of \eqref{eq_DPExpected}  becomes
\begin{align}\label{eq_opt_limit_1}
\mathsf{mse}_{\text{opt}}= \frac{1}{6f_{\max}} + o\left(\frac{1}{f_{\max}}\right).
\end{align}
\fi 
The sampling policy \eqref{eq_opt_solution1} 
was also obtained  in \cite{Basar2014} for the case that $ Y_i =  0$ for all $i$.

%

Similarly, if $\alpha\rightarrow  0$ or $f_{\max}\rightarrow 0$, the age-optimal sampling policy in \eqref{eq_thm2_1} and \eqref{eq_thm2} becomes  periodic sampling \eqref{eq_uniform} with $\beta = {1}/{f_{\max}}+ o(1/f_{\max})$, and the optimal value of \eqref{eq_age} is $\mathsf{mse}_{\text{age-opt}}= {1}/{(2f_{\max})}+ o(1/f_{\max})$. Therefore, 
\begin{align}\label{eq_ratio_MSE}
\lim_{\alpha\rightarrow  0} \frac{\mathsf{mse}_{\text{opt}}}{\mathsf{mse}_{\text{age-opt}}}= \lim_{f_{\max}\rightarrow 0}\frac{\mathsf{mse}_{\text{opt}}}{\mathsf{mse}_{\text{age-opt}}}=\frac{1}{3}.
\end{align}

\subsection{Long Service Time or Unbounded Sampling Rate}
If $\alpha\rightarrow\infty$ or $f_{\max}\rightarrow\infty$, 
\ifreport
as shown in Appendix \ref{app_scale}, 
\fi
the MSE-optimal sampling policy for solving \eqref{eq_DPExpected}
is given by \eqref{eq_opt_solution} where $\beta$ is determined by solving 
\begin{align} \label{eq_coro_1}
2\beta \mathbb{E}[\max(\beta,W_{Y}^2)]= {\mathbb{E}[\max(\beta^2,W_{Y}^4)]}{}.
\end{align}
Similarly, if $\alpha\rightarrow\infty$ or $f_{\max}\rightarrow\infty$, the age-optimal sampling policy for solving \eqref{eq_age} is given by \eqref{eq_thm2_1} where $\beta$ is determined by solving 
\begin{align} \label{eq_coro_2}
2\beta \mathbb{E}[\max(\beta,Y)] = {\mathbb{E}[\max(\beta^2,Y^2)]}{}.
\end{align}
In these limits, the ratio between ${\mathsf{mse}_{\text{opt}}}$ and ${\mathsf{mse}_{\text{age-opt}}}$ depends on the distribution of $Y$.

When the sampling rate is unbounded, i.e., $f_{\max}=\infty$, one logically reasonable policy is the zero-wait sampling policy  in \eqref{eq_Zero_wait} \cite{2015ISITYates,SunInfocom2016,report_AgeOfInfo2016,KaulYatesGruteser-Infocom2012}. This zero-wait sampling policy achieves the maximum throughput and the minimum queueing delay. However, this zero-wait sampling policy \emph{almost never} minimizes the MSE in \eqref{eq_DPExpected} and \emph{does not always} minimize the age of information in \eqref{eq_age}, as stated in the following two theorems:

We note that the zero-wait sampling policy can be expressed as \eqref{eq_opt_solution} with $\beta =0$. By checking when $\beta =0$ is satisfied in \eqref{eq_opt_solution} and \eqref{eq_thm1}, one can  obtain

\begin{theorem}\label{lem_zero_wait1}
Suppose that $f_{\max}=\infty$. Then, the zero-wait sampling policy \eqref{eq_Zero_wait} is an optimal solution to \eqref{eq_DPExpected} if and only if 
the service time $Y$ is equal to zero with probability one.
\end{theorem}
\ifreport
\begin{proof}
See Appendix \ref{app_zerowait}. 
\end{proof}
\fi

Hence, as long as the service time $Y$ has a small probability to be positive, the zero-waiting sampling policy is not an optimal solution to \eqref{eq_DPExpected}. Similarly, the optimality of zero-wait sampling policy for solving \eqref{eq_age} is characterized as

\begin{theorem}\label{lem_zero_wait2}\cite{report_AgeOfInfo2016}
Suppose that $f_{\max}=\infty$. Then, the zero-wait sampling policy \eqref{eq_Zero_wait} is an optimal solution to \eqref{eq_age} if and only if \emph{
\begin{align}\label{eq_zero_wait2}
\mathbb{E}[Y^2]\leq 2~ \text{ess}\inf Y~ \mathbb{E}[Y],
\end{align}}\!\!
where \emph{$\text{ess}\inf Y = \sup\{y\in[0,\infty): \Pr[Y< y]=0\}$} can be considered as the minimum possible value of $Y$.
\end{theorem}
\ifreport
\begin{proof} 
See Appendix \ref{app_zerowait}. 
\end{proof}
\fi




\section{Proof of Theorem \ref{thm_1}}\label{sec_proof}
We prove Theorem \ref{thm_1} in four steps: 
First, we  show that no sample should be generated when the server is busy, which simplifies the optimal online sampling problem. Second, we study the Lagrangian dual problem of the simplified problem, and decompose the Lagrangian dual problem into a series of \emph{mutually independent} per-sample control problems. Each of these per-sample control problems is a continuous-time Markov decision problem. Further, we utilize optimal stopping theory 
\cite{Shiryaev1978} 
to solve the per-sample control problems. Finally, we show that the Lagrangian duality gap of our Markov decision problem is zero. By this, Problem \eqref{eq_DPExpected} is solved. 
The details are as follows.

\subsection{Simplification of Problem \eqref{eq_DPExpected}} %
The following lemma is useful for simplifying \eqref{eq_DPExpected}.
\begin{lemma}\label{lem_zeroqueue}
Suppose that $\pi$ is a feasible  policy for Problem \eqref{eq_DPExpected}, in which at least one sample is taken when the server is busy processing an earlier generated sample. Then, there exists another feasible  policy $\pi_1$ for Problem \eqref{eq_DPExpected} that has a smaller estimation error than policy $\pi$.
Hence, it is suboptimal in Problem \eqref{eq_DPExpected} to take a new sample before the previous sample is delivered. 





%
\end{lemma}
\begin{proof}
Lemma \ref{lem_zeroqueue} is proven by using the strong Markov property of the Wiener process and the orthogonality principle of MMSE estimation. The details are provided in Appendix \ref{app_zeroqueue}. 
\end{proof}

By Lemma \ref{lem_zeroqueue}, we only need to consider a sub-class of sampling policies $\Pi_1\subset\Pi$ such that each  sample is generated and submitted to the server after the previous sample is delivered, i.e., 
\begin{align}
\Pi_1 
 &= \{\pi\in\Pi: S_{i} \geq D_{i-1} \text{ for all $i$}\}. 
\end{align}
This completely eliminates the waiting time wasted in the queue, and hence the queue is always kept empty. 
The \emph{information} that is available for determining $S_i$ includes the history of signal values $(W_t, t\in[0, S_i])$ and the service times  $(Y_1,\ldots, Y_{i-1})$ of  previous samples.\footnote{Note that the generation times $(S_1,\ldots, S_{i-1})$ of previous samples are also included in this information.} To characterize this statement precisely, let us define the $\sigma$-fields $\mathcal{F}_t = \sigma(W_s: s\in[0, t])$ and $\mathcal{F}_t^+ = \cap_{s>t}\mathcal{F}_s$. Then, $\{\mathcal{F}_t^+,t\geq0\}$ is the {filtration} (i.e., a non-decreasing and right-continuous  family of  $\sigma$-fields) of the Wiener process $W_t$.
Given the service times $(Y_1,\ldots, Y_{i-1})$ of previous samples, $S_{i} $ is a \emph{stopping time} with respect to the filtration $\{\mathcal{F}_t^+,t\geq0\}$ of the Wiener process $W_t$, that is
\begin{align}
[\{S_{i}\leq t\} | Y_1,\ldots, Y_{i-1}] \in \mathcal{F}_t^+,~\forall~t\geq0.\label{eq_stopping}
\end{align}  
Then, the policy space $\Pi_1$ can be alternatively expressed as
\begin{align}\label{eq_policyspace}
\!\!\!\!\Pi_1 = & \{S_i : [\{S_{i}\leq t\} | Y_1,\ldots, Y_{i-1}] \in \mathcal{F}_t^+,~\forall~t\geq0,\nonumber\\
&~~~~~~~S_{i}  \geq D_{i-1} \text{ for all $i$}, \nonumber\\
&~~~~~~~\text{$T_i = S_{i+1}-S_i$ is a regenerative process} \}.\!\!
\end{align}  
Recall that any policy in $\Pi$ satisfies ``$T_i = S_{i+1}-S_i$ is a regenerative process''.


Let $Z_i = S_{i+1} - D_{i}\geq0$ represent the \emph{waiting time} between the delivery time $D_{i}$ of sample $i$  and the generation time $S_{i+1}$ of sample $i+1$. Then,
$S_i =  Z_{0} +\sum_{j=1}^{i-1} (Y_{j} + Z_j)$ and $D_i =  \sum_{j=0}^{i-1} (Z_{j}+Y_{j+1})$. If $(Y_1,Y_2,\ldots)$ is given, $(S_0,S_1,\ldots)$ is uniquely determined by $(Z_0,Z_1,\ldots)$. Hence, one can also use $\pi = (Z_0,Z_1,\ldots)$ to represent a sampling policy. 
Because $T_i $ is a regenerative process, using the renewal theory \cite{Ross1996} and \cite[Section 6.1]{Haas2002}, one can show that in Problem \eqref{eq_DPExpected}, $\frac{1}{n} 
\mathbb{E}[S_n]$ is a convergent sequence and 
\begin{align}
&\limsup_{T\rightarrow \infty}\frac{1}{T}\mathbb{E}\left[\int_0^{T} (W_t-\hat W_t)^2dt\right] \nonumber
\\
=& \lim_{n\rightarrow \infty}\frac{\mathbb{E}\left[\int_0^{D_n} (W_t-\hat W_t)^2dt\right]}{\mathbb{E}[D_n]} \nonumber\\
=& \lim_{n\rightarrow \infty}\frac{\sum_{i=0}^{n-1}\mathbb{E}\left[\int_{D_{i}}^{D_{i+1}} (W_t- W_{S_{i}})^2dt\right]}{\sum_{i=0}^{n-1} \mathbb{E}\left[Y_i+Z_i\right]},\nonumber
\end{align}
where in the last step we have used $\mathbb{E}\left[D_n\right]=\mathbb{E}[\sum_{i=0}^{n-1} (Z_{i}+Y_{i+1})] = \mathbb{E}[\sum_{i=0}^{n-1} (Y_{i}+Z_{i})]$.
Hence, \eqref{eq_DPExpected} can be rewritten as the following Markov decision problem:
\begin{align}\label{eq_Simple}
{\mathsf{mse}}_{\text{opt}}\triangleq&\inf_{\pi\in\Pi_1} \lim_{n\rightarrow \infty}\frac{\sum_{i=0}^{n-1}\mathbb{E}\left[\int_{D_{i}}^{D_{i+1}} (W_t\!-\!W_{S_{i}})^2dt\right]}{\sum_{i=0}^{n-1} \mathbb{E}\left[Y_i+Z_i\right]} \\
&~\text{s.t.}~\lim_{n\rightarrow \infty} \frac{1}{n} 
\sum_{i=0}^{n-1} \mathbb{E}\left[Y_i+Z_i\right]\geq \frac{1}{f_{\max}},\label{eq_Simple_constraint}
\end{align}
where ${\mathsf{mse}}_{\text{opt}}$ is the optimal  value of \eqref{eq_Simple}.

In order to solve \eqref{eq_Simple}, let us consider the following Markov decision problem with a parameter $c\geq 0$:
\begin{align}\label{eq_SD}
\!\!p(c)\!\triangleq\!\inf_{\pi\in\Pi_1}&\lim_{n\rightarrow \infty}\frac{1}{n}\sum_{i=0}^{n-1}\!\mathbb{E}\!\left[\int_{D_{i}}^{D_{i+1}} \!\!\!\!(W_t-W_{S_{i}})^2dt\!-\! c(Y_i\!+\!Z_i)\!\right]\!\!\!\!\\
\text{s.t.}~&\lim_{n\rightarrow \infty} \frac{1}{n} 
\sum_{i=0}^{n-1} \mathbb{E}\left[Y_i+Z_i\right]\geq \frac{1}{f_{\max}},\nonumber
\end{align}
where $p(c)$ is the optimum  value of \eqref{eq_SD}. Similar with Dinkelbach's method  \cite{Dinkelbach67} for nonlinear fractional programming, we can obtain the following lemma for our Markov decision problem:
\begin{lemma} \label{lem_ratio_to_minus}
The following assertions are true:
\begin{itemize}
\vspace{0.5em}
\item[(a).] \emph{${\mathsf{mse}}_{\text{opt}} \gtreqqless c $} if and only if $p(c)\gtreqqless 0$. 
\vspace{0.5em}
\item[(b).] If $p(c)=0$, the solutions to \eqref{eq_Simple}
and \eqref{eq_SD} are identical. 
\end{itemize}
\end{lemma}
\begin{proof}
See Appendix \ref{app_ratio_to_minus}.
\end{proof}
Hence, the solution to \eqref{eq_Simple} can be obtained by solving \eqref{eq_SD} and seeking a ${\mathsf{mse}}_{\text{opt}}\geq 0$ such that 
\begin{align}\label{eq_c}
p({\mathsf{mse}}_{\text{opt}})=0. 
\end{align}

\subsection{Lagrangian Dual Problem of  \eqref{eq_SD} when \emph{$c = {\mathsf{mse}}_{\text{opt}}$} }

Although \eqref{eq_SD} is a continuous-time Markov decision problem with a continuous state space, rather than a convex optimization problem, it is possible to use the Lagrangian dual approach to solve \eqref{eq_SD} and show that it admits no duality gap. 

When $c = {\mathsf{mse}}_{\text{opt}}$, define the following Lagrangian
\begin{align}\label{eq_Lagrangian}
& L(\pi;\lambda) \nonumber\\
=& \lim_{n\rightarrow \infty}\!\frac{1}{n}\sum_{i=0}^{n-1}\mathbb{E}\!\left[\int_{D_{i}}^{D_{i+1}} \!\!\!\!(W_t\!-\!W_{S_{i}})^2dt\!-\! ({\mathsf{mse}}_{\text{opt}}+\lambda)(Y_i\!+\!Z_i)\right] \nonumber\\
&~+ \frac{\lambda}{f_{\max}},
\end{align}
where $\lambda\geq0$ is the dual variable. 
Let
\begin{align}\label{eq_primal}
e(\lambda) \triangleq \inf_{\pi\in\Pi_1}  L(\pi;\lambda).
\end{align}
Then, the Lagrangian dual problem of \eqref{eq_SD} is defined by
\begin{align}\label{eq_dual}
d \triangleq \max_{\lambda\geq 0}e(\lambda),
\end{align}
where $d$ is the optimum value of \eqref{eq_dual}. 
Weak duality \cite{Bertsekas2003,Boyd04} implies that $d \leq p({\mathsf{mse}}_{\text{opt}})$. 
In Section \ref{sec:dual}, we will establish strong duality, i.e., $d = p({\mathsf{mse}}_{\text{opt}})$.



In the sequel, we solve \eqref{eq_primal}.
%
Using the stopping times and martingale theory of the Wiener process, we can obtain the following  lemma: 
%
\begin{lemma}\label{lem_stop}
Let $\tau \geq0$ be a stopping time of the Wiener process $W_t$ with $\mathbb{E}[\tau^2]<\infty$, 
then
\begin{align}\label{eq_stop}
\mathbb{E}\left[\int_0^\tau W_t^2 dt\right]= \frac{1}{6}\mathbb{E}\left[ W_\tau^4 \right].
\end{align}
\end{lemma}
\begin{proof}
See Appendix \ref{applem_stop}.
\end{proof}

By using Lemma \ref{lem_stop} and  the sufficient statistics of \eqref{eq_primal},  we can show that for every $i = 1,2,\ldots$,
\begin{align}\label{eq_integral}
&\mathbb{E}\left[\int_{D_{i}}^{D_{i+1}} (W_t- W_{S_i})^2dt\right] \nonumber\\
= & \frac{1}{6}\mathbb{E}\left[ (W_{S_{i}+Y_i+Z_i}-W_{S_i})^4 \right] + \mathbb{E}\left[Y_i + Z_i \right] \mathbb{E}\left[Y_i \right],  
\end{align}
which is proven in Appendix \ref{app_integral}.

For any $s\geq 0$, define the $\sigma$-fields $\mathcal{F}^{s}_t = \sigma(W_{s+v}-W_{s}: v\in[0, t])$ and $\mathcal{F}^{s+}_t = \cap_{v>t} \mathcal{F}^{s}_v$, as well as the {filtration} $\{\mathcal{F}^{s+}_t,t\geq0\}$ of the {time-shifted Wiener process} $\{W_{s+t}-W_{s},t\in[0,\infty)\}$. Define ${\mathfrak{M}}_s$ as the set of square-integrable stopping times of $\{W_{s+t}-W_{s},t\in[0,\infty)\}$, i.e.,
\begin{align}
\mathfrak{M}_{s} = \{\tau \geq 0:  \{\tau\leq t\} \in \mathcal{F}^{s+}_t, \mathbb{E}\left[\tau^2\right]<\infty\}.\nonumber
\end{align}
By substituting  
  \eqref{eq_integral} into  \eqref{eq_primal}  and using again the sufficient statistics of \eqref{eq_primal},  
we can  obtain
\begin{theorem}\label{thm_solution_form}
An optimal solution  $(Z_0,Z_1,\ldots)$ 
to 
\eqref{eq_primal} 
satisfies 
\begin{align}\label{eq_opt_stopping}
Z_i =& \arg \inf_{\tau  \in \mathfrak{M}_{S_i+Y_i}} \mathbb{E}\left[\frac{1}{2}(W_{S_{i}+Y_i+\tau}-W_{S_i})^4\right.\nonumber\\
&~~~~~~~~~~~~~~~~\left.- {\beta}(Y_i+\tau)\bigg|W_{S_{i}+Y_i}-W_{S_i},Y_i \right],
\end{align}
where $\beta$ is given by \emph{
\begin{align}\label{eq_beta_new}
\beta = 3({\mathsf{mse}}_{\text{opt}} + \lambda - \mathbb{E}\left[Y \right]) \geq 0.
\end{align}}\!\!\!

\ignore{where $\beta \triangleq 3({\mathsf{mse}}_{\text{opt}} + \lambda - \mathbb{E}\left[Y \right]) \geq0$ is determined by solving
\begin{align}\label{eq_equation}
\mathbb{E}[Y_i+Z_i] \!=\! \max\left(\frac{1}{f_{\max}}, \frac{\mathbb{E}[(W_{S_{i}+Y_i+Z_i}-W_{S_i})^4]}{2\beta}\right).\!
\end{align}
The optimal  value of \eqref{eq_Simple} is then given by \emph{
\begin{align}\label{thm_3_obj}
\mathsf{mse}_{\text{opt}}=\frac{\mathbb{E}[(W_{S_{i}+Y_i+Z_i}-W_{S_i})^4]}{6\mathbb{E}[Y_i+Z_i] } + \mathbb{E}[Y].
\end{align}}}
\end{theorem}
\begin{proof}
See Appendix \ref{app_solution_form}. 
\end{proof}

Note that because the $Y_i$'s are \emph{i.i.d.} and the strong Markov property of the Wiener process, 
the $Z_i$'s as solutions of \eqref{eq_opt_stopping} are also \emph{i.i.d.} 

\subsection{Per-Sample Optimal Stopping Solution to \eqref{eq_opt_stopping}}
We use  optimal stopping theory  \cite{Shiryaev1978} 
to solve \eqref{eq_opt_stopping}.
Let us first pose \eqref{eq_opt_stopping} in the language of  optimal stopping.
A continuous-time two-dimensional Markov chain $X_t$ on a probability space $(\mathbb{R}^2,\mathcal{F},\mathbb{P})$ is defined as follows:  Given the initial state $ X_0=x = (s,b)$, the state $X_t$ at time $t$ is  
\begin{align}\label{eq_Markov}
X_t= (s + t, b + W_t), 
\end{align}
where $\{W_t,t\geq 0\}$ is a standard Wiener process. Define $\mathbb{P}_x(A) = \mathbb{P}(A | X_0 = x)$ and $\mathbb{E}_x Z = \mathbb{E}(Z | X_0 = x)$, respectively, 
as the conditional probability of event $A$ and the conditional expectation of random variable $Z$ 
for given  initial state $X_0 =x$. Define the $\sigma$-fields $\mathcal{F}^{X}_t  = \sigma(X_{v}: v\in[0, t])$ and $\mathcal{F}^{X+}_t  = \cap_{v>t} \mathcal{F}^{X}_v$, as well as the {filtration} $\{\mathcal{F}^{X+}_t ,t\geq0\}$ of the Markov chain $X_t$.
A random variable $\tau: \mathbb{R}^2\rightarrow [0,\infty)$ is said to be a \emph{stopping time} of $X_t$ if $\{\tau \leq t\}\in \mathcal{F}^{X+}_t $ {for all}~$t\geq0$. Let ${\mathfrak{M}}$ be the set of  square-integrable stopping times of $X_t$, i.e.,
\begin{align}
\mathfrak{M} = \{\tau \geq 0:  \{\tau\leq t\} \in \mathcal{F}^{X+}_t , \mathbb{E}\left[\tau^2\right]<\infty\}.\nonumber
\end{align}

Our goal is to solve the following optimal stopping problem:
\begin{align}\label{eq_stop_prob}
\sup_{\tau\in \mathfrak{M}}  \mathbb{E}_x g(X_\tau),
\end{align}
where $X_0 = x$ is the initial state of the Markov chain $X_t$, the function $g: \mathbb{R}^2\rightarrow \mathbb{R}$ is defined as 
\begin{align}\label{eq_fun}
g(s,b) =\beta s  - \frac{1}{2}b^4
\end{align}
with parameter $\beta \geq 0$. Notice that \eqref{eq_opt_stopping} is a special case of \eqref{eq_stop_prob} where  
 the initial state is $ x=(Y_i,W_{S_{i}+Y_i}-W_{S_i})$,  and  $W_t$ is replaced by the time-shifted Wiener process $W_{S_{i}+Y_i+t}-W_{S_i}$.

\begin{theorem}\label{thm_optimal_stopping}
For all $x= (s,b)\in \mathbb{R}^2$ and  $\beta\geq0$, an optimal stopping time for solving \eqref{eq_stop_prob} is
\begin{align}\label{eq_opt_stop_solution}
\tau^* = \inf \left\{ t \geq 0:\! \left|b+W_t\right| \!\geq\! \sqrt{\beta}\right\}.
\end{align}
\end{theorem} 


In order to prove Theorem \ref{thm_optimal_stopping}, let us define the function
\begin{align}
u(x) = \mathbb{E}_x g(X_{\tau^*}) 
\end{align}
and establish some properties of $u(x)$.

\begin{lemma} \label{lem1_stop}
 $u(x)\geq g(x)$ for all $x\in \mathbb{R}^2$, and 
\begin{align}
u(s,b) = \left\{\begin{array}{l l} 
\beta s- \frac{1}{2}b^4 , &\text{if}~ b^2 \geq \beta;
\vspace{0.5em}\\
\beta s+ \frac{1}{2} \beta ^2- \beta b^2 , &\text{if}~ b^2 < \beta.
\end{array}\right.
\end{align}
\end{lemma} 
\begin{proof}
See Appendix \ref{app_optimal_stopping}. 
\end{proof}

A function $f(x)$ is said to be \emph{excessive} for the process $X_t$ if \cite{Shiryaev1978}
\begin{align}
\mathbb{E}_x f(X_t)\leq f(x),~\text{for all}~ t\geq 0, ~x\in \mathbb{R}^2.
\end{align} 
By using the It\^{o}-Tanaka-Meyer formula \cite[Theorem 7.14 and Corollary 7.35]{BMbook10} in stochastic calculus, we can obtain
\begin{lemma}\label{lem2_stop}
The function $u(x)$ is excessive for the process $X_t$.
\end{lemma} 

\begin{proof}
See Appendix \ref{app_optimal_stopping1}. 
\end{proof}

Now, we are ready to prove Theorem \ref{thm_optimal_stopping}. 
\begin{proof}[Proof of Theorem \ref{thm_optimal_stopping}]
In Lemma \ref{lem1_stop} and Lemma \ref{lem2_stop}, we have shown that $u(x)=\mathbb{E}_x g(X_{\tau^*})$ is an excessive function and $u(x)\geq g(x)$. In addition, it is known that $\mathbb{P}_x(\tau^*<\infty) = 1$ for all $x\in \mathbb{R}^2$ \cite[Theorem 8.5.3]{Durrettbook10}. These conditions, together with the Corollary to Theorem 1 in \cite[Section 3.3.1]{Shiryaev1978}, imply that $\tau^*$ is an optimal stopping time of \eqref{eq_stop_prob}. This completes the proof.
\end{proof}

A consequence of Theorem \ref{thm_optimal_stopping} is
\begin{corollary}\label{coro_stop}
An optimal solution to \eqref{eq_opt_stopping} is
\begin{align}\label{eq_opt_stop_solution1}
Z_i = \inf \left\{ t \geq 0:\! \left|W_{S_i+Y_i+t}-W_{S_i}\right| \geq \sqrt{\beta}\right\}. 
\end{align}
In addition, this solution satisfies
\begin{align}
&\mathbb{E}[Y_i+Z_i] = \mathbb{E}[\max(\beta,W_Y^2)],\label{eq_expression0}\\
&\mathbb{E}[(W_{S_{i}+Y_i+Z_i}-W_{S_i})^4] = \mathbb{E}[\max(\beta^2,W_Y^4)].\label{eq_expression}
\end{align}
\begin{proof}
See Appendix \ref{app_coro_stop}. 
\end{proof}
\end{corollary}



\subsection{Zero Duality Gap between \eqref{eq_SD} and \eqref{eq_dual}} \label{sec:dual}

Strong duality is established in the following theorem:
\begin{theorem}\label{thm_zero_gap}
If \emph{$c= {\mathsf{mse}}_{\text{opt}}$}, the following assertions are true:
\begin{itemize}
\vspace{0.5em}
\item[(a).] The duality gap between \eqref{eq_SD} and \eqref{eq_dual} is zero, i.e., \emph{$d = p({\mathsf{mse}}_{\text{opt}})$}.
\vspace{0.5em}
\item[(b).] A common optimal solution to \eqref{eq_DPExpected},  \eqref{eq_Simple}, and \eqref{eq_SD} is given by \eqref{eq_opt_solution}  and \eqref{eq_thm1}. The optimal value of \eqref{eq_DPExpected} is given by \eqref{thm_1_obj}.
\end{itemize}
\end{theorem}

\begin{proof}[Proof Sketch of Theorem \ref{thm_zero_gap}]
We use \cite[Prop. 6.2.5]{Bertsekas2003} to find a \emph{geometric multiplier} \cite[Definition 6.1.1]{Bertsekas2003} for Problem \eqref{eq_SD}. This tells us that the duality gap between \eqref{eq_SD} and \eqref{eq_dual} must be zero, because otherwise there is no geometric multiplier \cite[Prop. 6.2.3(b)]{Bertsekas2003}.\footnote{Note that geometric multiplier is   different from the traditional Lagrangian multiplier.} This result holds  not only for convex optimization problem, but also for general non-convex optimization and Markov decision problems like \eqref{eq_SD}. 
See Appendix \ref{app_zero_gap} for the details.
\end{proof}
Hence, Theorem~\ref{thm_1} follows from Theorem \ref{thm_zero_gap}.


\section{Numerical Results}\label{sec:simulation}
In this section, we evaluate the estimation performance achieved by the following four sampling policies:  

\begin{itemize}
\item[1.] \emph{Periodic sampling:} The policy in \eqref{eq_uniform} with $\beta = f_{\max}$. 
\item[2.] \emph{Zero-wait sampling \cite{2015ISITYates,report_AgeOfInfo2016,SunInfocom2016,KaulYatesGruteser-Infocom2012}:} The sampling policy in \eqref{eq_Zero_wait}, which is feasible when $f_{\max}\geq 1/\EE[Y_i]$.

\item[3.] \emph{Age-optimal sampling \cite{report_AgeOfInfo2016,SunInfocom2016}:} The sampling policy in \eqref{eq_thm2_1} and \eqref{eq_thm2}, which is the optimal solution to  \eqref{eq_age}.

\item[4.] \emph{MSE-optimal sampling:} The sampling policy in \eqref{eq_opt_solution} and \eqref{eq_thm1}, which is the optimal solution to \eqref{eq_DPExpected}.
\end{itemize}
Let ${\mathsf{mse}_{\text{periodic}}}$, ${\mathsf{mse}_{\text{zero-wait}}}$, ${\mathsf{mse}_{\text{age-opt}}}$, and ${\mathsf{mse}_{\text{opt}}}$, 
be the MSEs of periodic sampling, zero-wait sampling, age-optimal sampling,  MSE-optimal sampling, respectively. 
According to \eqref{eq_compare_aware}, as well as the facts that periodic sampling is  feasible for \eqref{eq_age} and zero-wait sampling is  feasible for \eqref{eq_age} when $f_{\max}\geq 1/\EE[Y_i]$,  
we can obtain  
\begin{align}
&{\mathsf{mse}_{\text{opt}}}\leq{\mathsf{mse}_{\text{age-opt}}}\leq{\mathsf{mse}_{\text{periodic}}}, \nonumber\\
&{\mathsf{mse}_{\text{opt}}}\leq{\mathsf{mse}_{\text{age-opt}}}\leq{\mathsf{mse}_{\text{zero-wait}}},~\text{when $f_{\max}\geq \frac{1}{\EE[Y_i]}$,} \nonumber
\end{align}
which fit with our numerical results below.

\begin{figure}
\centering \includegraphics[width=0.35\textwidth]{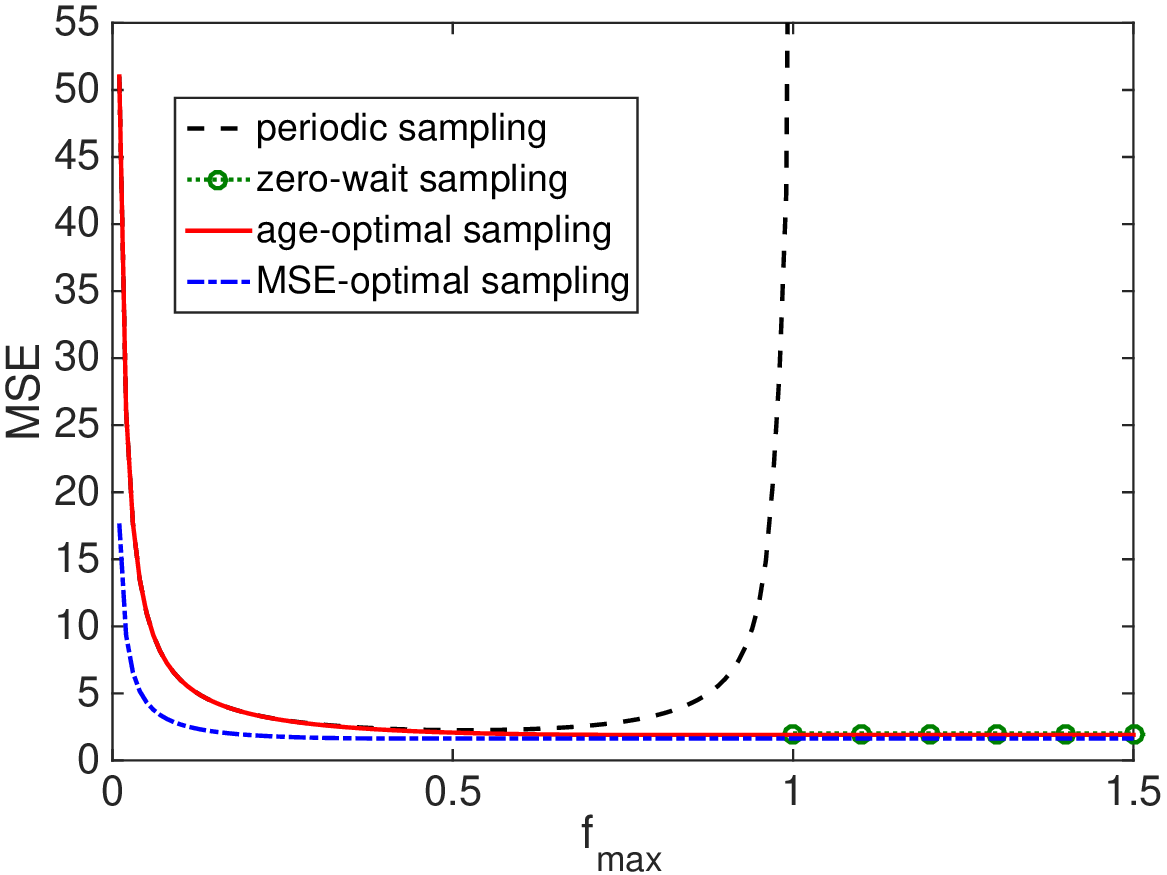} \caption{MSE vs. ${f_{\max}}$ tradeoff for \emph{i.i.d.} exponential service time and $\mathbb{E} [Y_i] = 1$, where  zero-wait sampling is not feasible when $f_{\max}< 1/\EE[Y_i]$  and hence is not plotted in that regime.}
\label{fig1} \vspace{-0.cm}
\end{figure}

\ifreport
\begin{figure}
\centering \includegraphics[width=0.35\textwidth]{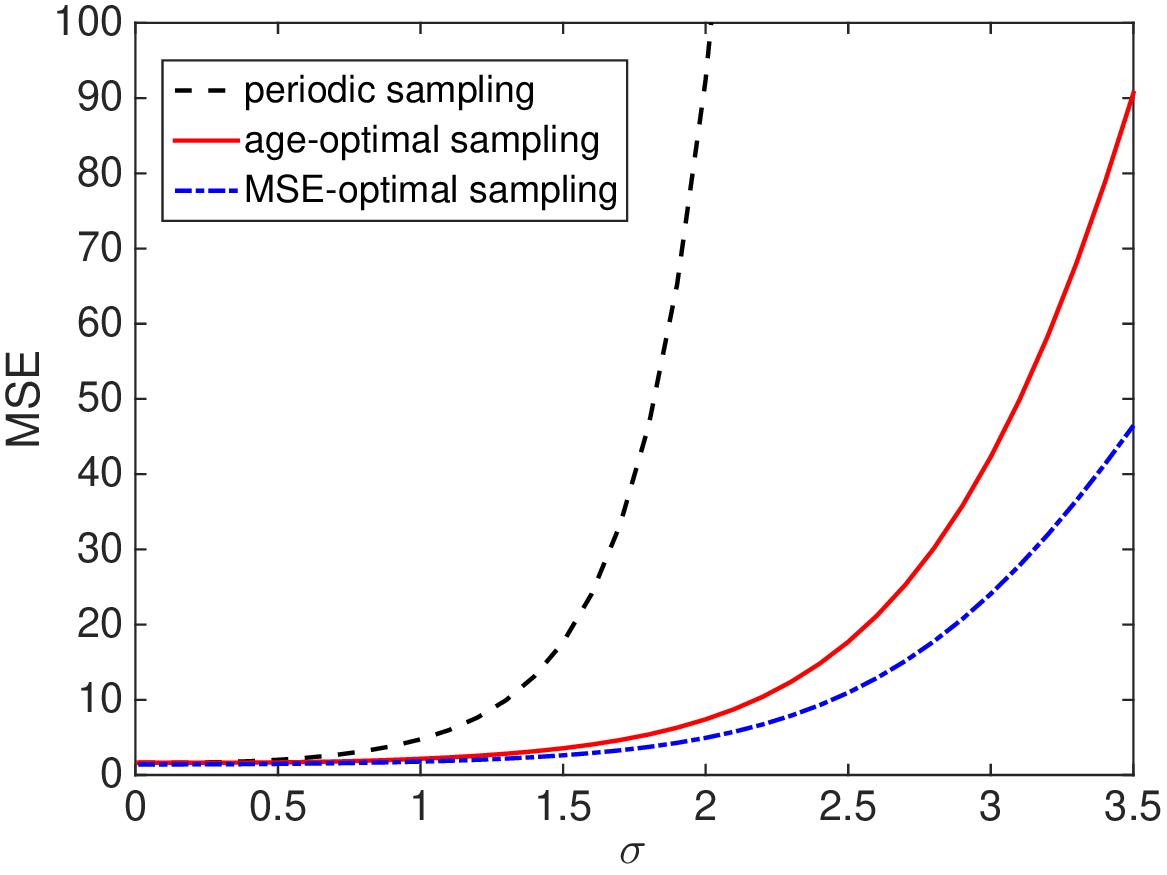} \caption{MSE vs. the scale parameter  $\sigma$ of \emph{i.i.d.} log-normal service time for $f_{\max}=0.8$ and $\mathbb{E} [Y_i] = 1$, where zero-wait sampling is not feasible because $f_{\max}< 1/\EE[Y_i]$. }
\label{fig2} \vspace{-0.cm}
\end{figure}
\else
\fi

Figure \ref{fig1} depicts the tradeoff between MSE and $f_{\max}$  for \emph{i.i.d.} exponential service time with mean $\mathbb{E}[Y_i] = 1/\mu=1$. Hence, the maximum throughput of the queue is $\mu=1$. In this setting, ${\mathsf{mse}_{\text{periodic}}}$ is characterized by eq. (25) of \cite{KaulYatesGruteser-Infocom2012}, which was obtained using a D/M/1 queueing model \cite{Kleinrock1976}. For small values of $f_{\max}$,   age-optimal sampling is similar with periodic sampling, and hence ${\mathsf{mse}_{\text{age-opt}}}$ and ${\mathsf{mse}_{\text{periodic}}}$ are of similar values. However, as $f_{\max}$ approaches the maximum throughput $1$, ${\mathsf{mse}_{\text{periodic}}}$ blows up to infinity. 
This is because the queue length in periodic sampling is large at high sampling frequencies, and the samples become stale during their long waiting times in the queue. On the other hand, ${\mathsf{mse}_{\text{opt}}}$ and ${\mathsf{mse}_{\text{age-opt}}}$ decrease with respect to $f_{\max}$. The reason is that the set of feasible policies satisfying the constraint in \eqref{eq_DPExpected} and \eqref{eq_age} becomes larger as $f_{\max}$ grows, and hence the optimal  values of \eqref{eq_DPExpected} and \eqref{eq_age} are decreasing in $f_{\max}$. Moreover, the gap between ${\mathsf{mse}_{\text{opt}}}$ and ${\mathsf{mse}_{\text{age-opt}}}$ is large for small values of $f_{\max}$. The ratio ${\mathsf{mse}_{\text{opt}}}/{\mathsf{mse}_{\text{age-opt}}}$ tends to $1/3$ as $f_{\max}\rightarrow 0$, which is in accordance with \eqref{eq_ratio_MSE}.
As we expected, ${\mathsf{mse}_{\text{zero-wait}}}$ is larger than ${\mathsf{mse}_{\text{opt}}}$ and ${\mathsf{mse}_{\text{age-opt}}}$ when $f_{\max}\geq1$. In summary, {periodic sampling is far from optimal if the sampling rate is too low or sufficiently high; age-optimal sampling is far from optimal if the sampling rate is too low.}


\ifreport

\begin{figure}
\centering \includegraphics[width=0.35\textwidth]{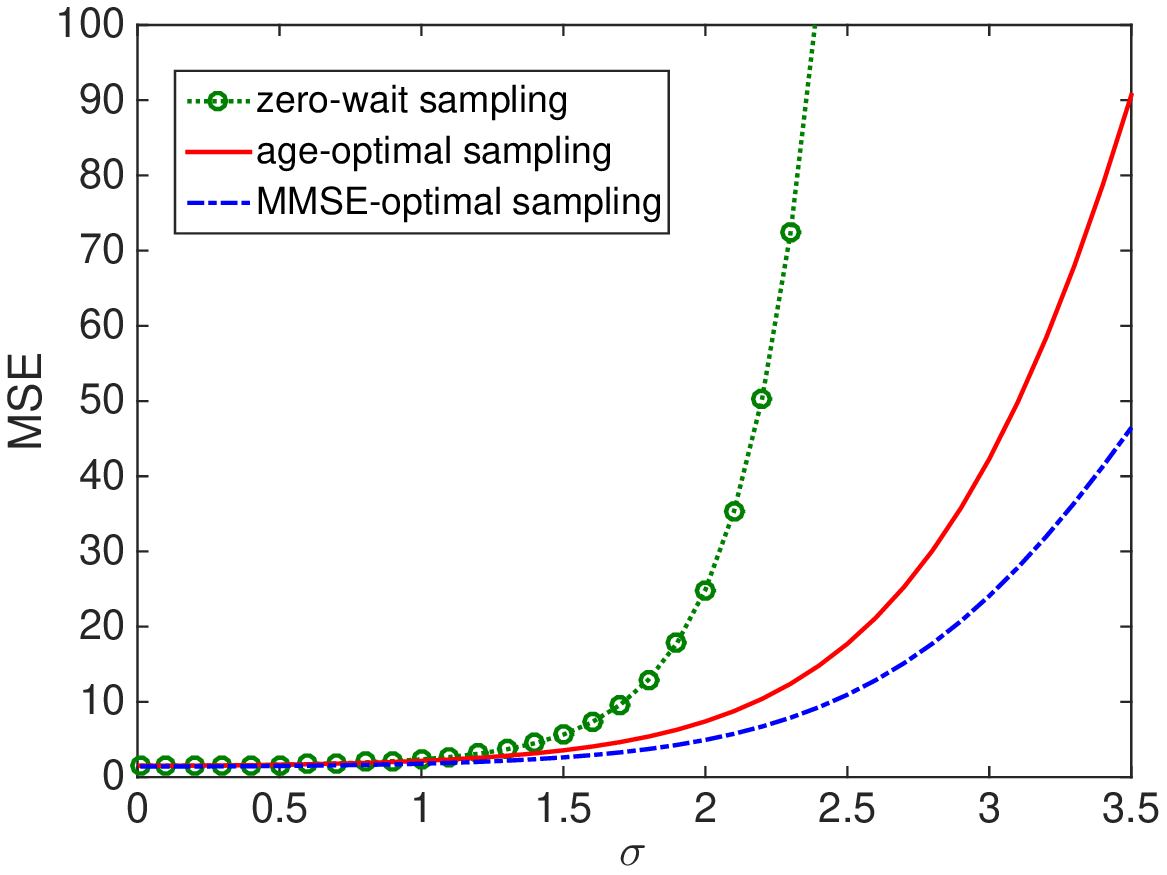} \caption{MSE vs. the scale parameter  $\sigma$ of \emph{i.i.d.} log-normal service time for $f_{\max}=1.5$ and $\mathbb{E} [Y_i] = 1$, where ${\mathsf{mse}_{\text{periodic}}}=\infty$ due to queueing.}
\label{fig3} \vspace{-0.cm}
\end{figure}

Figure \ref{fig2} and Figure \ref{fig3} illustrate the MSE of \emph{i.i.d.} log-normal service time for $f_{\max}=0.8$ and $f_{\max}=1.5$, respectively, where $Y_i = e^{\sigma X_i}/ \mathbb{E}[e^{\sigma X_i}]$, $\sigma>0$ is the scale parameter of log-normal distribution, and $(X_1,X_2,\ldots)$ are \emph{i.i.d.} Gaussian random variables with zero mean and unit variance. Because $ \mathbb{E}[Y_i] = 1$, the maximum throughput of the queue is  $1$. In Fig. \ref{fig2}, since $f_{\max}<1$, zero-wait sampling is not feasible and hence is not plotted. As the scale parameter $\sigma$ grows, the tail of the log-normal distribution becomes heavier and heavier. We observe that ${\mathsf{mse}_{\text{periodic}}}$ grows quickly with respect to $\sigma$, much faster than ${\mathsf{mse}_{\text{opt}}}$ and ${\mathsf{mse}_{\text{age-opt}}}$. 
In addition, the gap between ${\mathsf{mse}_{\text{opt}}}$ and ${\mathsf{mse}_{\text{age-opt}}}$ increases as $\sigma$ grows.
In Fig. \ref{fig3}, because $f_{\max}> 1$, ${\mathsf{mse}_{\text{periodic}}}$ is infinite and hence is not plotted. We can find that ${\mathsf{mse}_{\text{zero-wait}}}$ grows quickly with respect to $\sigma$ and is much larger than ${\mathsf{mse}_{\text{opt}}}$ and ${\mathsf{mse}_{\text{age-opt}}}$. 
In summary, {periodic sampling, age-optimal sampling, and zero-wait sampling policies are all far from optimal if the service times 
follow a heavy-tail distribution}.

\fi

\section{Conclusion}\label{conclusion}
In this paper, we have investigated optimal sampling of the Wiener process for remote estimation over a queue. 
The  optimal sampling policy for minimizing the mean square estimation error subject to an average sampling rate constraint has been obtained in a semi-closed form. 
We prove that a threshold-based sampler is optimal and the optimal threshold is found exactly.
Analytical and numerical comparisons with several important sampling policies, including age-optimal sampling, zero-wait sampling, and traditional periodic sampling,  have been provided. 
The results in this paper generalize recent research on age of information by adding a signal-based control model, and generalize existing studies on remote estimation by adding a queueing model with random service times.

\section*{Acknowledgement}
The  authors appreciate Ness B. Shroff and Roy D. Yates for their careful reading of the conference version  of this paper and their valuable suggestions. The authors are also grateful to Aditya Mahajan for pointing out an error in an earlier version of this paper. 


\appendices

\section{Proof of \eqref{eq_esti}}\label{app_estimation}

\ignore{
One can also use the technique in the proof of \eqref{eq_age_MSE}. Find the optimal estimate in the interval $[D_i\wedge T,D_{i+1}\wedge T]$. The estimate is a stochastic process, calculus of variations could be needed, and because it is not functional optimization, calculus of variations may or may not be sufficient. In any case, this is an alternative proof.  
}
We use the calculus of variations to prove \eqref{eq_esti}. Define $x\wedge y = \min\{x,y\}$. 
Let us consider a functional $h$ of the estimate $\hat W_t$, which is defined as 
\begin{align}\label{eq_functional_h}
h (\hat W_t) &= \mathbb{E}\left\{\int_{D_{i}\wedge T}^{D_{i+1}\wedge T} (\hat W_t-W_{t})^2dt\Bigg|(S_j, W_{S_j}, D_j)_{j\leq i}\right\}
\end{align}
 for any $T>0$.
By using Lemma 4 in \cite{report_AgeOfInfo2016}, it is not hard to show that $h (\hat W_t)$ is a convex functional of the estimate $\hat W_t$. In the sequent, we will find the optimal estimate that solves
\begin{align}\label{eq_optimal_estimate_problem}
\min_{\hat W_t} h(\hat W_t).
\end{align}

Let $f_t$ and $g_t$ be two estimates, which are functions of the information available at the estimator $\{S_i, W_{S_i}, D_i: D_i \leq t\}$.
Similar to the one-sided sub-gradient in finite dimensional space, the one-sided G\^ateaux derivative of the functional $h$ in the direction of $g$ at a point $f$ is given by   
\begin{align}
& \delta h(f; g) \nonumber\\
=& \lim_{\epsilon\rightarrow 0^+} \frac{h (f_t + \epsilon g_t)-h (f_t )} {\epsilon} \nonumber\\
=&  \lim_{\epsilon\rightarrow 0^+} \frac{1}{\epsilon}\mathbb{E}\left\{\int_{D_{i}\wedge T}^{D_{i+1}\wedge T} (f_t + \epsilon g_t-W_{t})^2 - (f_t -W_{t})^2dt\right.\nonumber \\
& ~~~~~~~~~~~~~~~~~~~~~~~~~\Bigg|(S_j, W_{S_j}, D_j)_{j\leq i}\Bigg\}.\nonumber\\
=&  \lim_{\epsilon\rightarrow 0^+} \mathbb{E}\left\{\int_{D_{i}\wedge T}^{D_{i+1}\wedge T} 2(f_t -W_{t}) g_t +  \epsilon g_t^2dt\right.\nonumber \\
& ~~~~~~~~~~~~~~~~~~~~~~~~~\Bigg|(S_j, W_{S_j}, D_j)_{j\leq i}\Bigg\}\nonumber
\end{align}
\begin{align}
=& \mathbb{E}\left\{\int_{D_{i}\wedge T}^{D_{i+1}\wedge T}  2(f_t -W_{t}) g_t d t ~\Bigg|(S_j, W_{S_j}, D_j)_{j\leq i}\right\}\nonumber\\
=& \mathbb{E}\left\{\int_{D_{i}\wedge T}^{D_{i+1}\wedge T}  2\left(f_t -\mathbb{E}\left[W_{t}|(S_j, W_{S_j}, D_j)_{j\leq i}\right]\right) g_t d t \right.\nonumber \\ & ~~~~~~~~~~~~~~~~~~~~~~~~~\Bigg|(S_j, W_{S_j}, D_j)_{j\leq i}\Bigg\},\label{eq_derivative}
\end{align}
where the last step follows from the iterated law of expectations. 
According to \cite[p. 710]{Bertsekas}, $f_t$ is an optimal solution to \eqref{eq_optimal_estimate_problem} if and only if
\begin{align}
\delta h(f; g) \geq 0,~~\forall ~g.\nonumber
\end{align}
By $\delta h(f; g)=- \delta h(f; -g)$, we get
\begin{align}\label{eq_gradient111}
\delta h(f; g) = 0,~~\forall ~g.
\end{align}
Since $g_t$ is arbitrary, by \eqref{eq_derivative} and \eqref{eq_gradient111}, the optimal solution to \eqref{eq_optimal_estimate_problem} is 
\begin{align}
f_t =& \mathbb{E}[W_{t}|(S_j, W_{S_j}, D_j)_{j\leq i}],\nonumber\\
=&  W_{S_{i}}+\mathbb{E}[W_{t}-W_{S_{i}}|(S_j, W_{S_j}, D_j)_{j\leq i}] \nonumber\\
&~~~~~~~~~~~~~~~~~~~~~~ t\in[D_{i}\wedge T, D_{i+1}\wedge T). 
\end{align}
Notice that under any online sampling policy $\pi$, $\{S_{j}, W_{S_{j}},D_{j}, $ $  j \leq i\}$ are determined by the source $(W_t,t\in[0, S_{i}])$ and the service times $(Y_1,\ldots,Y_i)$. 
According to (i) the strong Markov property of the Wiener process \cite[Theorem 2.16 and Remark 2.17]{BMbook10} and (ii) the fact that the $Y_i$'s are independent of the Wiener process $W_t$, we obtain that  for any given realization of $(S_{j}, W_{S_{j}},D_{j})_{ j \leq i}$, $\{W_{t} - W_{S_{i}}, t\geq S_i\}$ is a Wiener process. Hence, 
\begin{align}
\mathbb{E}[W_{t}-W_{S_{i}}|(S_j, W_{S_j}, D_j)_{j\leq i}] = 0 
\end{align}
for all $t\geq S_i$. Therefore, the optimal solution to \eqref{eq_optimal_estimate_problem}
is 
\begin{align}\label{eq_estimate11111}
f_t = W_{S_i}, ~\text{if}~t\in[D_{i}\wedge T, D_{i+1}\wedge T), i = 1,2,\ldots
\end{align}

Finally, we note that 
\begin{align}\label{eq_age_MSE3}
&\lim_{n\rightarrow \infty}\sum_{i=0}^n\mathbb{E}\left\{\int_{D_{i}\wedge T}^{D_{i+1}\wedge T} (W_t-\hat W_t)^2dt\right\}\nonumber\\
\overset{}{=} & \lim_{n\rightarrow \infty} \mathbb{E}\left\{\int_0^{D_{n}\wedge T} (W_t-\hat W_t)^2dt\right\}\nonumber\\
\overset{(a)}{=}& \mathbb{E}\left\{\lim_{n\rightarrow \infty}\int_0^{D_{n}\wedge T} (W_t-\hat W_t)^2dt\right\}\nonumber\\
\overset{(b)}{=}&\mathbb{E}\left\{\int_0^{ T} (W_t-\hat W_t)^2dt\right\} 
\end{align}
where in Step (a) we have used the monotonic convergence theorem, and Step (a) is due to $\lim_{n\rightarrow \infty} D_n =\infty$ almost surely, which was obtained in \eqref{eq_infinite}. Hence, the MMSE estimation problem can be formulated as
\begin{align}\label{eq_age_MSE4}
&\min_{\hat W_t}\limsup_{T\rightarrow \infty}\frac{1}{T}\mathbb{E}\left\{\int_0^{ T} (W_t-\hat W_t)^2dt\right\} \nonumber\\
=&\min_{\hat W_t}\limsup_{T\rightarrow \infty}\lim_{n\rightarrow \infty}\frac{1}{T}\sum_{i=0}^n\mathbb{E}\left\{\int_{D_{i}\wedge T}^{D_{i+1}\wedge T} (W_t-\hat W_t)^2dt\right\}.
\end{align}
Recall that \eqref{eq_estimate11111} is the solution of \eqref{eq_functional_h} for any $T>0$. Let $T \rightarrow \infty$ in \eqref{eq_estimate11111}, we obtain that \eqref{eq_esti} is the MMSE estimator for solving \eqref{eq_age_MSE4}. This completes the proof. 

\section{Proof of \eqref{eq_age_MSE}}\label{app_age}
If $\pi$ is independent of $\{W_t,t\in[0,\infty)\}$, the $S_i$'s and $D_i$'s are independent of $\{W_t,t\in[0,\infty)\}$. Define $x\wedge y = \min\{x,y\}$. For any $T>0$, let us consider the term
\begin{align}
&\mathbb{E}\left\{\int_{D_{i}\wedge T}^{D_{i+1}\wedge T} (W_t-\hat W_t)^2dt\right\} \nonumber
\end{align}
in the following two cases:

\emph{Case 1:} If $D_{i}\wedge T\geq S_i$, we can obtain
\begin{align}\label{eq_age_MSE1}
&\mathbb{E}\left\{\int_{D_{i}\wedge T}^{D_{i+1}\wedge T} (W_t-\hat W_t)^2dt\right\} \nonumber\\
\overset{}{=}&\mathbb{E}\left\{\int_{D_{i}\wedge T}^{D_{i+1}\wedge T} (W_t-\hat W_{S_i})^2dt\right\} \nonumber\\
\overset{(a)}{=}& \mathbb{E}\left\{\mathbb{E}\left\{\int_{D_{i}\wedge T}^{D_{i+1}\wedge T} (W_t-W_{S_i})^2dt\Bigg|S_i,D_i,D_{i+1}\right\} \right\} \nonumber\\
\overset{(b)}{=}& \mathbb{E}\left\{\int_{D_{i}\wedge T}^{D_{i+1}\wedge T} \mathbb{E}\left\{(W_t-W_{S_i})^2|S_i,D_i,D_{i+1}\right\} dt\right\} \nonumber\\
\overset{(c)}{=}& \mathbb{E}\left\{\int_{D_{i}\wedge T}^{D_{i+1}\wedge T}  (t - S_i) dt\right\} \nonumber\\
\overset{(d)}{=}& \mathbb{E}\left\{\int_{D_{i}\wedge T}^{D_{i+1}\wedge T} \Delta(t) dt\right\}, 
\end{align} 
where Step (a) is due to the law of iterated expectations, Step (b) is due to Fubini's theorem, Step (c) is due to  the strong Markov property of the Wiener process \cite[Theorem 2.16]{BMbook10} and the fact that $S_i,D_i,D_{i+1}$ are  independent of the Wiener process,
 and Step (d) is due to \eqref{eq_age_def}. 

\emph{Case 2:} If $D_{i}\wedge T<S_i$, then the fact $D_{i}\geq S_i$ implies that $T<S_i \leq D_i \leq D_{i+1}$. Hence, $D_i \wedge T= D_{i+1}\wedge T = T$ and   
\begin{align}
&\mathbb{E}\left\{\int_{D_{i}\wedge T}^{D_{i+1}\wedge T} (W_t-\hat W_t)^2dt\right\}= \mathbb{E}\left\{\int_{D_{i}\wedge T}^{D_{i+1}\wedge T} \Delta(t) dt\right\} = 0.\nonumber
\end{align} 
Therefore, \eqref{eq_age_MSE1} holds in both cases.

By using an argument similar to \eqref{eq_age_MSE3}, we can obtain
\begin{align}\label{eq_age_MSE2}
&\lim_{n\rightarrow \infty}\sum_{i=0}^n\mathbb{E}\left\{\int_{D_{i}\wedge T}^{D_{i+1}\wedge T} \Delta(t) dt\right\}=\mathbb{E}\left\{\int_0^{ T} \Delta(t) dt\right\}.
\end{align} 
Combining \eqref{eq_age_MSE3}-\eqref{eq_age_MSE2}, \eqref{eq_age_MSE} is proven.

\section{Proofs of \eqref{eq_betaasy} and  \eqref{eq_opt_limit_1}}\label{app_low_delay}
If $f_{\max}\rightarrow 0$,  \eqref{eq_thm1} tells us that
\begin{align}
\mathbb{E}[\max(\beta,W_Y^2)] =\frac{1}{f_{\max}},
\end{align}
which implies
\begin{align}
\beta \leq \frac{1}{f_{\max}} \leq \beta+ \mathbb{E}[W_Y^2] = \beta+ \mathbb{E}[Y].
\end{align}
Hence,
\begin{align}
\frac{1}{f_{\max}} -\mathbb{E}[Y] \leq \beta \leq \frac{1}{f_{\max}}.
\end{align}
If $f_{\max}\rightarrow 0$, \eqref{eq_betaasy} follows. 

Because $Y$ is independent of the Wiener process, using the law of iterated expectations and the Gaussian distribution of the Wiener process, we can obtain $\mathbb{E}[W_Y^4]=3\mathbb{E}[Y^2]$ and $\mathbb{E}[W_Y^2]=3\mathbb{E}[Y]$. Hence, 
\begin{align}
\beta\leq &~\mathbb{E}[\max(\beta,W_Y^2)] \leq \beta+ \mathbb{E}[W_Y^2]=\beta+ \mathbb{E}[Y],\nonumber\\
\beta^2\leq &~\mathbb{E}[\max(\beta^2,W_Y^4)] \leq \beta^2+ \mathbb{E}[W_Y^4]=\beta^2+ 3\mathbb{E}[Y^2].\nonumber
\end{align}
Therefore, 
\begin{align}\label{eq_eq_betaasy}
\frac{\beta^2}{\beta+ \mathbb{E}[Y]}\leq \frac{\mathbb{E}[\max(\beta^2,W_Y^4)]}{\mathbb{E}[\max(\beta,W_Y^2)]}\leq \frac{\beta^2+ 3\mathbb{E}[Y^2]}{\beta}.
\end{align}
By combining  \eqref{thm_1_obj}, \eqref{eq_betaasy}, and \eqref{eq_eq_betaasy}, \eqref{eq_opt_limit_1} follows in the case of $f_{\max}\rightarrow 0$.

If $\alpha\rightarrow  0$, then $Y \rightarrow 0$ and $W_Y \rightarrow 0$ with probability one. Hence, $\mathbb{E}[\max(\beta,W_Y^2)] \rightarrow \beta$ and $\mathbb{E}[\max(\beta^2,W_Y^4)] \rightarrow \beta^2$. Substituting these
 into \eqref{eq_thm1} and \eqref{eq_eq_betaasy}, yields
\begin{align}
\lim_{\alpha\rightarrow  0 }\beta = \frac{1}{f_{\max}},~ \lim_{\alpha\rightarrow  0 }\left\{\frac{\mathbb{E}[\max(\beta^2,W_Y^4)]}{6\mathbb{E}[\max(\beta,W_Y^2)]} + \mathbb{E}[Y]\right\} = \frac{1}{6f_{\max}}.\nonumber
\end{align}
By this, \eqref{eq_betaasy} and \eqref{eq_opt_limit_1} are proven in the case of $\alpha\rightarrow 0$.
This completes the proof.

\section{Proof of \eqref{eq_coro_1}}\label{app_scale}
If $f_{\max}\rightarrow\infty$, the sampling rate constraint in \eqref{eq_DPExpected} can be removed. By \eqref{eq_thm1}, the optimal $\beta$ is determined by  \eqref{eq_coro_1}.

If $\alpha\rightarrow\infty$, let us consider the equation
\begin{align} \label{eq_a_equation}
\mathbb{E}[\max(\beta,W_Y^2)] \!=\! \frac{\mathbb{E}[\max(\beta^2,W_Y^4)]}{2\beta}.
\end{align}
If $Y$ grows by $\alpha$ times, then $\beta$ and $\mathbb{E}[\max(\beta,W_Y^2)]$ in \eqref{eq_a_equation} both should grow by $\alpha$ times, and $\mathbb{E}[\max(\beta^2,W_Y^4)]$ in \eqref{eq_a_equation} should grow by $\alpha^2$ times. Hence, if $\alpha\rightarrow\infty$, it holds in \eqref{eq_thm1} that 
\begin{align}
\frac{1}{f_{\max}}\leq\frac{\mathbb{E}[\max(\beta^2,W_Y^4)]}{2\beta}
\end{align} and 
the solution to \eqref{eq_thm1} is given by \eqref{eq_coro_1}. This completes the proof.

\section{Proofs of Theorems \ref{lem_zero_wait1} and \ref{lem_zero_wait2}}\label{app_zerowait}
\begin{proof}[Proof of Theorem \ref{lem_zero_wait1}]
The zero-wait policy can be expressed as \eqref{eq_opt_solution} with $\beta =0$. Because $Y$ is independent of the Wiener process, using the law of iterated expectations and the Gaussian distribution of the Wiener process, we can obtain $\mathbb{E}[W_Y^4]=3\mathbb{E}[Y^2]$. 
According to \eqref{eq_coro_1}, $\beta=0$  if and only if $\mathbb{E}[W_Y^4]=3\mathbb{E}[Y^2]=0$ which is equivalent to $Y=0$ with probability one. This completes the proof.
\end{proof}
\begin{proof}[Proof of Theorem \ref{lem_zero_wait2}]
In the one direction, the zero-wait policy can be expressed as \eqref{eq_thm2_1} with $\beta \leq \text{ess}\inf Y$. If the zero-wait policy is optimal, then  the solution to \eqref{eq_coro_2} must satisfy $\beta \leq \text{ess}\inf Y$, which further implies $\beta\leq Y$ with probability one. From this, we can get 
\begin{align}
 2  \text{ess}\inf Y \mathbb{E}[Y]\geq 2\beta \mathbb{E}[Y] = {\mathbb{E}[Y^2]},
\end{align}
By this, \eqref{eq_zero_wait2} follows.

In the other direction, if \eqref{eq_zero_wait2} holds, we will show that the zero-wait policy is age-optimal by considering the following two cases. 

\emph{Case 1:} $\mathbb{E}[Y]> 0$. By choosing \begin{align}\label{eq_coro_5}
\beta = \frac{\mathbb{E}[Y^2]}{2\mathbb{E}[Y]},
\end{align}we can get $\beta \leq \text{ess}\inf Y$ from \eqref{eq_zero_wait2} 
and hence
\begin{align}\label{eq_coro_6}
\beta \leq  Y
\end{align}
with probability one.
According to \eqref{eq_coro_5} and \eqref{eq_coro_6}, such a $\beta$ is the solution to \eqref{eq_coro_2}. Hence, the zero-wait policy expressed by \eqref{eq_thm2_1} with $\beta \leq \text{ess}\inf Y$ is the age-optimal policy. 

\emph{Case 2:} $\mathbb{E}[Y]= 0$ and hence $Y=0$ with probability one. In this case, $\beta =0$ is the solution to \eqref{eq_coro_2}. Hence, the zero-wait policy expressed by \eqref{eq_thm2_1} with $\beta =0$ is the age-optimal policy. 

Combining these two cases, the proof is completed.
\end{proof}

\section{Proof of Lemma \ref{lem_zeroqueue}}\label{app_zeroqueue}


Recall that $S_i$ and and $G_i$ are the sampling time and the service starting time of sample $i$, respectively.
Suppose that in the sampling policy $\pi$, sample $i$ is generated when the server is busy sending another sample, and hence sample $i$  
needs to wait for some time before being submitted to the server, i.e., $S_i<G_i$. 
Let us consider a \emph{virtual} sampling policy $\pi' = \{S_0,\ldots, S_{i-1},G_i, S_{i+1},\ldots\}$ such that the generation time of sample $i$ is postponed from $S_i$ to $G_i$. We call policy $\pi'$ a virtual policy because it may happen that  $G_i > S_{i+1}$. 
However, this will not affect our  proof below. We will show that the MSE of the sampling policy $\pi'$ is smaller than that of the sampling policy $\pi =\{S_0,\ldots, S_{i-1},S_i, S_{i+1},\ldots\}$. 

Note that the Wiener process $\{W_t: t\in[0, \infty)\}$ does not change according to the sampling policy, and the sample delivery times $\{D_0,D_1, D_2,\ldots\}$ remain the same in policy $\pi$ and policy $\pi'$. 
Hence, the only difference between policies $\pi$ and $\pi'$ is that \emph{the generation time of sample $i$ is postponed from $S_i$ to $G_i$}. 
The MMSE estimator under policy $\pi$ is given by \eqref{eq_esti} and the  MMSE estimator under policy $\pi'$ is given by
\begin{align}\label{eq_esti_pi1}
\hat W_t = &\mathbb{E}[W_{t}|(S_j, W_{S_j}, D_j)_{j\leq i-1}, (G_i, W_{G_i}, D_i)]\nonumber\\
 = &\left\{\begin{array}{l l} 0,& t\in[0,D_1);\\
W_{G_i},& t\in[D_i,D_{i+1}); \\
W_{S_j},& t\in[D_j,D_{j+1}),~j\neq i, j\geq 1. \\
\end{array}\right.
\end{align}

Next, we consider a third virtual sampling policy $\pi''$ in which the samples $(W_{G_i}, G_i)$ and  $(W_{S_i}, S_i)$ are both delivered to the estimator at time $D_i$. Clearly, the estimator under policy $\pi''$ has more information than those under policies $\pi$ and $\pi'$. By following the arguments in Appendix \ref{app_estimation}, one can show that the MMSE estimator under policy $\pi''$ is
\begin{align}\label{eq_esti_pi2}
\hat W_t =& \mathbb{E}[W_{t}|(S_j, W_{S_j}, D_j)_{j\leq i}, (G_i, W_{G_i}, D_i)] \nonumber\\
=&\left\{\begin{array}{l l} 0,& t\in[0,D_1);\\
W_{G_i},& t\in[D_i,D_{i+1}); \\
W_{S_j},& t\in[D_j,D_{j+1}),~j\neq i, j\geq 1. \\
\end{array}\right.
\end{align}
Notice that, because of the strong Markov property of Wiener process, the estimator under policy $\pi''$ uses the fresher sample $W_{G_i}$, instead of the stale sample $W_{S_i}$, to construct $\hat W_t$ during $[D_i,D_{i+1})$. Because the estimator under policy $\pi''$ has more information than that under policy $\pi$, one can imagine that policy $\pi''$ has a smaller estimation error than policy $\pi$, i.e., 
for any $T>0$
\begin{align}\label{eq_small_error}
&\mathbb{E}\left\{\int_{D_i\wedge T} ^{D_{i+1}\wedge T} (W_t-W_{S_i})^2dt\right\}\nonumber\\
\geq &\mathbb{E}\left\{\int_{D_i\wedge T} ^{D_{i+1}\wedge T} (W_t-W_{G_i})^2  dt\right\}.
\end{align}
To prove \eqref{eq_small_error}, we invoke the orthogonality principle of the MMSE estimator \cite[Prop. V.C.2]{Poor:1994} under policy $\pi''$ and obtain 
\begin{align}\label{eq_orthogonality}
\mathbb{E}\left\{\int_{D_i\wedge T} ^{D_{i+1}\wedge T}2 (W_t-W_{G_i}) (W_{G_i}- W_{S_i}) dt\right\}=0,
\end{align}
where we have used the fact that  $W_{G_i}$ and $W_{S_i}$ are available by the MMSE estimator under policy $\pi''$.
Next, from \eqref{eq_orthogonality}, we can get
\begin{align}
&\mathbb{E}\left\{\int_{D_i\wedge T} ^{D_{i+1}\wedge T} (W_t-W_{S_i})^2dt\right\} \nonumber\\
= & \mathbb{E}\left\{\int_{D_i\wedge T} ^{D_{i+1}\wedge T} (W_t-W_{G_i})^2 +(W_{G_i}- W_{S_i})^2 dt\right\} \nonumber\\
& + \mathbb{E}\left\{\int_{D_i\wedge T} ^{D_{i+1}\wedge T}2 (W_t-W_{G_i}) (W_{G_i}- W_{S_i}) dt\right\} \nonumber \\
= & \mathbb{E}\left\{\int_{D_i\wedge T} ^{D_{i+1}\wedge T} (W_t-W_{G_i})^2 +(W_{G_i}- W_{S_i})^2 dt\right\} \nonumber\\
\geq &\mathbb{E}\left\{\int_{D_i\wedge T} ^{D_{i+1}\wedge T} (W_t-W_{G_i})^2  dt\right\}. 
\end{align}
In other words, the estimation error of policy $\pi''$ is no greater than that of policy $\pi$. Furthermore, by comparing \eqref{eq_esti_pi1} and \eqref{eq_esti_pi2}, we can see  that the MMSE estimators under policies $\pi''$ and $\pi'$ are exact the same. Therefore, the estimation error of policy $\pi'$ is no greater than that of policy $\pi$.

By repeating the above arguments for all samples $i$ satisfying $S_i<G_i$, one can show that the sampling policy $\pi_1 = \{S_0,G_1,\ldots, G_{i-1},G_i, G_{i+1},\ldots\}$ is better than the sampling policy $\pi =\{S_0,S_1,\ldots, S_{i-1},S_i, S_{i+1},\ldots\}$.
This completes the proof.

\section{Proof of Lemma  \ref{lem_ratio_to_minus}}\label{app_ratio_to_minus}
Part (a) is proven in  two steps:

\emph{Step 1:} We will prove that {$\mathsf{mse}_{\text{opt}} \leq c $ if and only if $p(c)\leq 0$}.

If $\mathsf{mse}_{\text{opt}} \leq c $, then there exists a policy $\pi= (Z_0,Z_1,\ldots)\in\Pi_1$ that is feasible for both \eqref{eq_Simple}
and \eqref{eq_SD}, which satisfies  
\begin{align}\label{eq_step1}
\lim_{n\rightarrow \infty}\frac{\sum_{i=0}^{n-1}\mathbb{E}\left[\int_{D_{i}}^{D_{i+1}} (W_t-W_{S_{i}})^2dt\right]}{ \sum_{i=0}^{n-1} \mathbb{E}\left[Y_i\!+\!Z_i\right]}\leq c.
\end{align}
Hence,
\begin{align}\label{eq_step2}
\lim_{n\rightarrow \infty}\frac{\frac{1}{n}\sum_{i=0}^{n-1}\mathbb{E}\left[\int_{D_{i}}^{D_{i+1}} (W_t-W_{S_{i}})^2dt - c(Y_i+Z_i)\right]}{\frac{1}{n} \sum_{i=0}^{n-1} \mathbb{E}\left[Y_i\!+\!Z_i\right]}\leq0.
\end{align}
Because the inter-sampling times $T_i=Y_i+Z_i$ are regenerative, $\mathbb{E}[{k_{j+1}}-{k_j}]<\infty$ and $0<\mathbb{E}[(S_{k_{j+1}}-S_{k_j})^2]<\infty$ for all $j$, the renewal theory \cite{Ross1996} tells us that the limit 
$\lim_{n\rightarrow \infty}\frac{1}{n} \sum_{i=0}^{n-1} \mathbb{E}\left[Y_i\!+\!Z_i\right]$ exists and is positive.  By this, we get 
\begin{align}\label{eq_step3}
\lim_{n\rightarrow \infty}\frac{1}{n}\sum_{i=0}^{n-1}\mathbb{E}\left[\int_{D_{i}}^{D_{i+1}} (W_t-W_{S_{i}})^2dt - c(Y_i+Z_i)\right] \leq 0.
\end{align}
Therefore, $p(c)\leq 0$.

On the reverse direction, if $p(c)\leq 0$, then there exists a policy $\pi= (Z_0,Z_1,\ldots)\in\Pi_1$ that is feasible for both \eqref{eq_Simple}
and \eqref{eq_SD}, which satisfies \eqref{eq_step3}. Because the limit 
$\lim_{n\rightarrow \infty}\frac{1}{n} \sum_{i=0}^{n-1} \mathbb{E}\left[Y_i\!+\!Z_i\right]$ exists and is positive, from \eqref{eq_step3}, we can derive \eqref{eq_step2} and \eqref{eq_step1}. Hence, $\mathsf{mse}_{\text{opt}} \leq c $. By this, we have proven that {$\mathsf{mse}_{\text{opt}} \leq c $ if and only if $p(c)\leq 0$}. 

\emph{Step 2:} We needs to prove that $\mathsf{mse}_{\text{opt}} < c $ if and only if $p(c)< 0$. This statement can be proven by using the arguments in \emph{Step 1}, in which ``$\leq$'' should be replaced by ``$<$''. Finally, from the statement of \emph{Step 1}, it immediately follows that $\mathsf{mse}_{\text{opt}} > c $ if and only if $p(c)> 0$. This completes the proof of part (a). 

Part (b): We first show that each optimal solution to \eqref{eq_Simple} is  an optimal solution to  \eqref{eq_SD}. 
By the claim of part (a), $p(c)= 0$ is equivalent to $\mathsf{mse}_{\text{opt}} = c $. Suppose that policy $\pi= (Z_0,Z_1,\ldots)\in\Pi_1$ is an optimal solution to  \eqref{eq_Simple}.
Then, $\mathsf{mse}_{\pi} = \mathsf{mse}_{\text{opt}} = c $. Applying this in the arguments of \eqref{eq_step1}-\eqref{eq_step3}, we can show that policy $\pi$ satisfies
\begin{align}
\lim_{n\rightarrow \infty}\frac{1}{n}\sum_{i=0}^{n-1}\mathbb{E}\left[\int_{D_{i}}^{D_{i+1}} (W_t-W_{S_{i}})^2dt - c(Y_i+Z_i)\right] = 0.\nonumber
\end{align}   
This and $p(c)= 0$ imply that  policy $\pi$ is an optimal solution to  \eqref{eq_SD}. 

Similarly, we can prove that each optimal solution to \eqref{eq_SD} is  an optimal solution to  \eqref{eq_Simple}. By this, part (b) is proven.



\section{Proof of Lemma \ref{lem_stop}}\label{applem_stop}
According to Theorem 2.51  of \cite{BMbook10}, 
$W_t^4 - 6\int_0^t W_s^2 ds$ 
is an martingale  of the Wiener process $\{W_t,t\in[0,\infty)\}$. 
Because the minimum of two stopping times is a stopping time and constant times are stopping times \cite{Durrettbook10}, it follows that $t\wedge \tau$ is a bounded stopping time for every $t\in[0,\infty)$, where $x\wedge y = \min\{x,y\}$. Then, it follows from Theorem 8.5.1 of \cite{Durrettbook10} that for every $t\in[0,\infty)$
\begin{align}
\!\!\!\!\mathbb{E}\!\left[\int_0^{t \wedge \tau }\!\!\! W_s^2 ds\right] &\!=\!  \frac{1}{6}\mathbb{E}\left[ W_{t \wedge \tau}^4 \right]. \label{eq_martigale1}
\end{align}
Notice that $\int_0^{t \wedge \tau} W_s^2 ds$ is positive and increasing with respect to $t$. By applying the monotone convergence theorem \cite[Theorem 1.5.5]{Durrettbook10}, we can obtain 
\begin{align}
\lim_{t\rightarrow \infty }\mathbb{E}\left[\int_0^{t \wedge \tau} W_s^2 ds\right] &= \mathbb{E}\left[\int_0^{\tau} W_s^2 ds\right].
\end{align}
Hence, the limit $\lim_{t\rightarrow \infty } \mathbb{E}\left[ W_{t \wedge \tau}^4 \right]$ exists. The remaining task is to show that 
\begin{align}
\lim_{t\rightarrow \infty } \mathbb{E}\left[ W_{t \wedge \tau}^4 \right] = \mathbb{E}\left[ W_{\tau}^4 \right]. \label{eq_goal}
\end{align}
Towards this goal, let us consider
\begin{align}
\!\!&\mathbb{E}\left[ W_{\tau}^4 \right]\nonumber\\
=& \mathbb{E}\left\{ [W_{t\wedge \tau}-(W_{\tau}-W_{t \wedge \tau})]^4 \right\}\nonumber\\
=& \mathbb{E}\left[W_{t\wedge \tau}^4\right]+4 \mathbb{E}\left[W_{t\wedge \tau}^3(W_{\tau}-W_{t \wedge \tau})\right]\nonumber\\
&+6 \mathbb{E}\left[W_{t\wedge \tau}^2(W_{\tau}-W_{t \wedge \tau})^2\right]\nonumber\\
&+4 \mathbb{E}\left[W_{t\wedge \tau}(W_{\tau}-W_{t \wedge \tau})^3\right]+\mathbb{E}\left[(W_{\tau}-W_{t \wedge \tau})^4\right]\nonumber\\
=& \mathbb{E}\left[W_{t\wedge \tau}^4\right]+4 \mathbb{E}\left[W_{t\wedge \tau}^3\right]\mathbb{E}\left[W_{\tau}-W_{t \wedge \tau}\right]\nonumber\\
&+6 \mathbb{E}\left[W_{t\wedge \tau}^2\right]\mathbb{E}\left[(W_{\tau}-W_{t \wedge \tau})^2\right]\nonumber\\
&+4 \mathbb{E}\left[W_{t\wedge \tau}\right]\mathbb{E}\left[(W_{\tau}-W_{t \wedge \tau})^3\right]\!+\!\mathbb{E}\left[(W_{\tau}-W_{t \wedge \tau})^4\right]\!,\!\!\nonumber
\end{align}
where in the last step we have used the strong Markov property of the Wiener process \cite[Theorem 2.16]{BMbook10}. By Wald's lemma for Wiener process \cite[Theorem 2.44 and Theorem 2.48]{BMbook10}, $\mathbb{E}\left[ W_{\tau}\right] = 0$ and $\mathbb{E}\left[ W_{\tau}^2\right] =  \mathbb{E}\left[\tau\right] $ for any stopping time $\tau$ with $\mathbb{E}\left[\tau\right] <\infty$. Hence,
\begin{align}
\mathbb{E}\left[W_{\tau}-W_{t \wedge \tau}\right]=0,\\
\mathbb{E}\left[W_{t\wedge \tau}\right]=0,
\end{align}
which implies 
\begin{align}
\mathbb{E}\left[ W_{\tau}^4 \right] =& \mathbb{E}\left[W_{t\wedge \tau}^4\right]+6 \mathbb{E}\left[W_{t\wedge \tau}^2\right]\mathbb{E}\left[(W_{\tau}-W_{t \wedge \tau})^2\right]\nonumber\\
&+\mathbb{E}\left[(W_{\tau}-W_{t \wedge \tau})^4\right]\nonumber\\
\geq&\mathbb{E}\left[W_{t\wedge \tau}^4\right],
\end{align}
and hence
\begin{align}\label{eq_fatou1}
\mathbb{E}\left[ W_{\tau}^4 \right] \geq \lim_{t\rightarrow \infty}\mathbb{E}\left[W_{t\wedge \tau}^4\right].
\end{align}
On the other hand, by Fatou's lemma \cite[Theorem 1.5.4]{Durrettbook10},
\begin{align}\label{eq_fatou}
 \mathbb{E}\left[ W_{\tau}^4 \right] = 
\mathbb{E}\left[ \liminf_{t\rightarrow \infty } W_{t \wedge \tau}^4 \right]
\leq \lim_{t\rightarrow \infty }  \mathbb{E}\left[ W_{t \wedge \tau}^4 \right] . 
\end{align} 
Combining \eqref{eq_fatou1} and \eqref{eq_fatou}, yields \eqref{eq_goal}. This completes the proof.

\section{Proof of  \eqref{eq_integral}}\label{app_integral}
The following lemma is needed in the proof of  \eqref{eq_integral}:
\begin{lemma}\label{lem_decompose}
For any $\lambda \geq 0$, there exists an optimal solution $(Z_0,Z_1,\ldots)$ to \eqref{eq_primal} in which $Z_i$ is independent of $(W_t, t\in[0,{S_i}])$ for all $i=1,2,\ldots$
\end{lemma}
\begin{proof}
Because the $Y_i$'s are \emph{i.i.d.}, $Z_i$ is independent of $Y_{i+1}, Y_{i+2},\ldots$, and the strong Markov property of the Wiener process \cite[Theorem 2.16]{BMbook10}, in the Lagrangian $L(\pi;\lambda)$ the term related to $Z_i$ is
\begin{align}\label{eq_decomposed_term}
&\mathbb{E}\left[\int_{S_{i}+Y_i}^{S_{i}+Y_i+Z_i +Y_{i+1}}\!\!\!\! (W_t-W_{S_{i}})^2dt\!-\! ({\mathsf{mse}}_{\text{opt}}+\lambda)(Y_i+Z_i)\right]\!\!,
\end{align}
which is determined by 
the control decision $Z_i$ and the recent information of the system $\mathcal{I}_i = (Y_i, (W_{S_{i}+t}-W_{S_{i}},$ $t \geq 0))$. According to \cite[p. 252]{Bertsekas2005bookDPVol1} and \cite[Chapter 6]{Kumar1986}, $\mathcal{I}_i$ is a \emph{sufficient statistics} for determining  $Z_i$ in \eqref{eq_primal}. Therefore, there exists an optimal policy $(Z_0,Z_1,\ldots)$ in which $Z_i$ is determined based on only $\mathcal{I}_i$, which is independent of $(W_t: t\in[0,S_i])$. 
This completes the proof.
\end{proof}

\begin{proof}[Proof of \eqref{eq_integral}]
By using \eqref{eq_policyspace} and Lemma \ref{lem_decompose}, we obtain that for
given $Y_i$ and $Y_{i+1}$, $Y_i$ and $Y_i+Z_i+Y_{i+1}$ are stopping times of the time-shifted Wiener process $\{W_{S_{i}+t}-W_{S_{i}},t\geq0\}$. Hence,

\begin{align}\label{eq_Q_i}
{=}& \mathbb{E}\left\{\int_{D_{i}}^{D_{i+1}} (W_t-W_{S_i})^2dt\right\} \nonumber\\
\overset{}{=} & \mathbb{E}\left\{\int_{Y_i}^{Y_i+Z_i+Y_{i+1}} (W_{S_{i}+t}-W_{S_{i}})^2dt\right\} \nonumber\\
\overset{(a)}{=} & \mathbb{E}\left\{\mathbb{E}\left\{\int_{Y_i}^{Y_i+Z_i+Y_{i+1}} (W_{S_{i}+t}-W_{S_{i}})^2dt\Bigg| Y_i, Y_{i+1}\right\}\right\} \nonumber\\
\overset{(b)}{=} &  \frac{1}{6}\mathbb{E}\left\{\mathbb{E}\left\{(W_{S_{i}+Y_i+Z_i+Y_{i+1}}-W_{S_{i}})^4\Bigg| Y_i, Y_{i+1}\right\}\!\!\right\} \nonumber\\
&-\frac{1}{6} \mathbb{E}\left\{\mathbb{E}\left\{(W_{S_{i}+Y_i}-W_{S_{i}})^4\Bigg| Y_i, Y_{i+1}\right\}\!\!\right\} \nonumber\\
\overset{(c)}{=} & \frac{1}{6}\mathbb{E}\left[ (W_{S_{i}+Y_i+Z_i+Y_{i+1}}\!-\!W_{S_{i}})^4 \right] - \frac{1}{6}\mathbb{E}\left[ (W_{S_{i}+Y_i}\!-\!W_{S_{i}})^4 \right]\!\!, 
\end{align}
where   Step (a) and Step (c) are due to the law of iterated expectations, and Step (b) is due to Lemma \ref{lem_stop}. 
Because $S_{i+1} = {S_{i}+Y_i+Z_i}$, we have
\begin{align}
&\mathbb{E}\left[ (W_{S_{i}+Y_i+Z_i+Y_{i+1}}-W_{S_{i}})^4 \right] \nonumber\\
=&\mathbb{E}\left\{[ (W_{S_{i}+Y_i+Z_i}-W_{S_i}) +(W_{S_{i+1}+Y_{i+1}}-W_{S_{i+1}})]^4 \right\} \nonumber\\
=&\mathbb{E}\left[(W_{S_{i}+Y_i+Z_i}-W_{S_i})^4\right] \nonumber\\
&+ 4\mathbb{E}\left[(W_{S_{i}+Y_i+Z_i}-W_{S_i})^3(W_{S_{i+1}+Y_{i+1}}-W_{S_{i+1}})\right]  \nonumber\\
&+6\mathbb{E}\left[ (W_{S_{i}+Y_i+Z_i}-W_{S_i})^2(W_{S_{i+1}+Y_{i+1}}-W_{S_{i+1}})^2\right]   \nonumber\\
&+ 4\mathbb{E}\left[(W_{S_{i}+Y_i+Z_i}-W_{S_i})(W_{S_{i+1}+Y_{i+1}}-W_{S_{i+1}})^3\right] \nonumber\\
&+ \mathbb{E}\left[(W_{S_{i+1}+Y_{i+1}}-W_{S_{i+1}})^4\right] \nonumber\\
=&\mathbb{E}\left[(W_{S_{i}+Y_i+Z_i}-W_{S_i})^4\right] \nonumber\\
&+ 4\mathbb{E}\left[(W_{S_{i}+Y_i+Z_i}-W_{S_i})^3\right]\mathbb{E}\left[(W_{S_{i+1}+Y_{i+1}}-W_{S_{i+1}})\right]  \nonumber\\
&+6\mathbb{E}\left[ (W_{S_{i}+Y_i+Z_i}-W_{S_i})^2\right] \mathbb{E}\left[(W_{S_{i+1}+Y_{i+1}}-W_{S_{i+1}})^2\right]   \nonumber\\
&+ 4\mathbb{E}\left[(W_{S_{i}+Y_i+Z_i}-W_{S_i})\right]\mathbb{E}\left[(W_{S_{i+1}+Y_{i+1}}-W_{S_{i+1}})^3\right] \nonumber\\
&+ \mathbb{E}\left[(W_{S_{i+1}+Y_{i+1}}-W_{S_{i+1}})^4\right], 
\end{align}
where in the last equation we have used the fact that $Y_{i+1}$ is independent of $Y_i$ and $Z_i$, and the strong Markov property of the Wiener process \cite[Theorem 2.16]{BMbook10}. 
By Wald's lemma for Wiener process \cite[Theorem 2.44 and Theorem 2.48]{BMbook10}, $\mathbb{E}\left[ W_{\tau}\right] = 0$ and $\mathbb{E}\left[ W_{\tau}^2\right] =  \mathbb{E}\left[\tau\right] $ for any stopping time $\tau$ with $\mathbb{E}\left[\tau\right] <\infty$. Hence,
\begin{align}
&\mathbb{E}\left[(W_{S_{i+1}+Y_{i+1}}-W_{S_{i+1}})\right] =0,\\
&\mathbb{E}\left[(W_{S_{i}+Y_i+Z_i}-W_{S_i})\right] =0,\\
&\mathbb{E}\left[ (W_{S_{i}+Y_i+Z_i}-W_{S_i})^2\right] \mathbb{E}\left[(W_{S_{i+1}+Y_{i+1}}-W_{S_{i+1}})^2\right] \nonumber\\
&= \mathbb{E}[Y_i+Z_i]\mathbb{E}[Y_{i+1}].
 \end{align}
Therefore, we have
\begin{align}
&\mathbb{E}\left[ (W_{S_{i}+Y_i+Z_i+Y_{i+1}}-W_{S_{i}})^4 \right] \nonumber\\
=&\mathbb{E}\left[(W_{S_{i}+Y_i+Z_i}-W_{S_i})^4\right]\! +\!6\mathbb{E}\left[Y_i\!+\!Z_i\right] \mathbb{E}\left[Y_{i+1}\right]\!\nonumber\\
&+ \!\mathbb{E}\left[(W_{S_{i+1}+Y_{i+1}}-W_{S_{i+1}})^4\right]\!.  
\end{align}
Finally, because $(W_{S_{i}+t}-W_{S_{i}})$ and $(W_{S_{i+1}+t}-W_{S_{i+1}})$ are both Wiener processes, and the $Y_i$'s are \emph{i.i.d.}, 
\begin{align}\label{eq_Q_i1}
\mathbb{E}\left[(W_{S_{i}+Y_{i}}-W_{S_{i}})^4\right] = \mathbb{E}\left[(W_{S_{i+1}+Y_{i+1}}\!-\!W_{S_{i+1}})^4\right].
\end{align}
Combining  \eqref{eq_Q_i}-\eqref{eq_Q_i1}, yields \eqref{eq_integral}.
\end{proof}

\section{Proof of Theorem  \ref{thm_solution_form}}\label{app_solution_form}


By \eqref{eq_integral}, \eqref{eq_decomposed_term}  can be rewritten as 

\begin{align}\label{eq_decomposed_term1}
&\mathbb{E}\left[\int_{S_{i}+Y_i}^{S_{i}+Y_i+Z_i +Y_{i+1}} \!\!\!\!\!\!(W_t-W_{S_{i}})^2dt\!-\! ({\mathsf{mse}}_{\text{opt}}+\lambda)(Y_i+Z_i)\right]\nonumber\\
=&\mathbb{E}\left[\frac{1}{6}(W_{S_{i}+Y_i+Z_i}-W_{S_i})^4\!-\! ({\mathsf{mse}}_{\text{opt}}+\lambda - \mathbb{E}[Y])(Y_i\!+\!Z_i)\right] \nonumber\\
=&\mathbb{E}\left[\frac{1}{6}(W_{S_{i}+Y_i+Z_i}-W_{S_i})^4\!-\! \frac{\beta}{3}(Y_i\!+\!Z_i)\right] \nonumber\\
=&\mathbb{E}\left[\frac{1}{6}[(W_{S_{i}+Y_i}\!-\!W_{S_{i}}) \!+\! (W_{S_{i}+Y_i+Z_i}\!-\!W_{S_{i}+Y_i})]^4\right.\nonumber\\
&- \frac{\beta}{3}(Y_i \!+\! Z_i)\bigg].
\end{align}
Because the $Y_i$'s are \emph{i.i.d.}  and the strong Markov property of the Wiener process \cite[Theorem 2.16]{BMbook10}, the expectation in \eqref{eq_decomposed_term1} is determined by the control decision $Z_i$ and the information $\mathcal{I}_i' = (W_{S_{i}+Y_i}-W_{S_{i}}, Y_i, (W_{S_{i}+Y_i+t}-W_{S_{i}+Y_i},$ $t \geq 0))$. According to \cite[p. 252]{Bertsekas2005bookDPVol1} and \cite[Chapter 6]{Kumar1986}, $\mathcal{I}_i'$ is a \emph{sufficient statistics} for determining the waiting time $Z_i$ in \eqref{eq_primal}. Therefore, there exists an optimal policy $(Z_0,Z_1,\ldots)$ in which $Z_i$ is determined based on only $\mathcal{I}_i'$. 
By this, \eqref{eq_primal} is decomposed into a sequence of per-sample control problems \eqref{eq_opt_stopping}. Combining \eqref{eq_Simple}, \eqref{eq_integral}, and Lemma \ref{lem_ratio_to_minus}, yields ${\mathsf{mse}}_{\text{opt}}\geq \mathbb{E}[Y]$. Hence, $\beta\geq 0$. 

We note that, because the $Y_i$'s are \emph{i.i.d.} and the strong Markov property of the Wiener process, the $Z_i$'s in this optimal policy are  \emph{i.i.d.} Similarly, the $(W_{S_{i}+Y_i+Z_i}-W_{S_i})$'s in this optimal policy are  \emph{i.i.d.}

\section{Proof of Lemma  \ref{lem1_stop}}\label{app_optimal_stopping}
Case 1: If $b^2 \geq \beta$, then \eqref{eq_opt_stop_solution} tells us that 
\begin{align}\label{eq_zero}
\tau^* = 0
\end{align}
 and 
\begin{align}
u(x) = \mathbb{E}[g(X_0)|X_0=x] = g(x) = \beta s  - \frac{1}{2}b^4.
\end{align}
Case 2: If $b^2 < \beta$, then $\tau^* > 0$ and $(b+ W_{\tau^*})^2 = \beta$. Invoking Theorem 8.5.5 in \cite{Durrettbook10}, yields 
\begin{align}\label{eq_expectation}
\mathbb{E}_x \tau^* = - (\sqrt{\beta} - b)(-\sqrt{\beta} -b) = \beta - b^2. 
\end{align}
Using this, we can obtain
\begin{align}
u(x) &= \mathbb{E}_x g(X_{\tau^*}) \nonumber\\
    &= \beta (s + \mathbb{E}_x \tau^*) - \frac{1}{2} \mathbb{E}_x \left[(b+ W_{\tau^*})^4\right] \nonumber\\
& = \beta (s + \beta - b^2)  - \frac{1}{2} \beta^2\nonumber\\
& = \beta s + \frac{1}{2}\beta^2 - b^2\beta.
\end{align}
Hence, in Case 2,
\begin{align}
u(x)-g(x)=\frac{1}{2}\beta^2 - b^2\beta + \frac{1}{2}b^4 =\frac{1}{2} (b^2-\beta)^2\geq 0.
\end{align}
By combining these two cases, Lemma \ref{lem1_stop} is proven. 

\section{Proof of Lemma  \ref{lem2_stop}}\label{app_optimal_stopping1}

The function $u(s,b)$ is continuous differentiable in $(s,b)$. In addition, $\frac{\partial^2}{\partial^2b}u(s,b)$ is continuous everywhere but at $b = \pm\sqrt{\beta}$. 
By the It\^{o}-Tanaka-Meyer formula \cite[Theorem 7.14 and Corollary 7.35]{BMbook10}, we obtain that almost surely 
\begin{align}\label{eq_ito}
&u(s+t, b+W_t) - u(s, b) \nonumber\\
=& \int_0^t \frac{\partial}{\partial b} u(s+r,b+W_{r}) dW_{r}  \nonumber\\
  &+ \int_0^t \frac{\partial}{\partial s} u(s+r,b+W_{r}) d r \nonumber\\
  &+ \frac{1}{2}\int_{-\infty}^\infty L^{a}(t)\frac{\partial^2}{\partial b^2} u(s+r,b+a)da , 
\end{align}
where $L^a(t)$ is the local time that the Wiener process spends at the level $a$, i.e.,
\begin{align}
L^a(t) = \lim_{\epsilon \downarrow 0} \frac{1}{2\epsilon}\int_0^t 1_{\{| W_s -a|\leq \epsilon\}} ds,
\end{align}
and $1_A$ is the indicator function of event $A$.
By the property of local times of the Wiener process \cite[Theorem 6.18]{BMbook10}, we obtain that almost surely
\begin{align}\label{eq_ito}
&u(s+t, b+W_t) - u(s, b) \nonumber\\
=& \int_0^t \frac{\partial}{\partial b} u(s+r,b+W_{r}) dW_{r}  \nonumber\\
  &+ \int_0^t \frac{\partial}{\partial s} u(s+r,b+W_{r}) d r \nonumber\\
  &+ \frac{1}{2}\int_0^t \frac{\partial^2 }{\partial b^2} u(s+r,b+W_{r}) d r. 
 \end{align}
 
Because
 \begin{align}
 \frac{\partial}{\partial b} u(s,b)  = \left\{\begin{array}{l l} 
-2b^3, &\text{if}~ b^2 \geq \beta;
\vspace{0.5em}\\
- 2\beta b, &\text{if}~ b^2 < \beta,
\end{array}\right.
\end{align}
we can obtain that for all $t\geq 0$ and all $x = (s,b)\in \mathbb{R}^2$ 
\begin{align}
&\mathbb{E}_x\left\{\int_0^t \left[\frac{\partial}{\partial b} u(s+r,b+W_{r})\right]^2 dr \right\} <\infty.
\end{align}
This and Theorem 7.11 of \cite{BMbook10} imply that $\int_0^t \frac{\partial}{\partial b} u(s+r,b+W_{r}) dW_{r}$ is a martingale and 
\begin{align}\label{eq_ito_integral}
\mathbb{E}_x\left[\int_0^t \frac{\partial}{\partial b} u(s+r,b+W_{r}) dW_{r}\right] =0,~ \forall~ t\geq0.
 \end{align}   
By combining \eqref{eq_Markov}, \eqref{eq_ito}, and \eqref{eq_ito_integral}, we get
\begin{align} \label{eq_difference}
\mathbb{E}_x \left[u(X_t)\right] \!-\! u(x) = \mathbb{E}_x \left\{\!\int_0^t \!\left[\frac{\partial}{\partial s} u(X_r) \!+\!\frac{1}{2} \frac{\partial^2 }{\partial b^2} u(X_r)\right]\!dr\right\}.
\end{align} 
It is easy to compute that if $b^2> \beta$, 
\begin{align} 
\frac{\partial}{\partial s} u(s,b) + \frac{1}{2}\frac{\partial^2 }{\partial b^2} u(s,b) =\beta -  3b^2 \leq 0; 
 \end{align}
and if $b^2< \beta$, 
\begin{align} 
\frac{\partial}{\partial s} u(s,b) + \frac{1}{2}\frac{\partial^2 }{\partial b^2} u(s,b) =\beta -  \beta = 0.
 \end{align} 
 Hence, 
 \begin{align} \label{eq_gradient}
 \frac{\partial}{\partial s} u(s,b) + \frac{1}{2}\frac{\partial^2 }{\partial b^2} u(s,b) \leq 0
  \end{align} 
 for all $(s,b)\in \mathbb{R}^2$ except for $b =  \pm\sqrt{\beta}$. Since the Lebesgue measure of those $r$ for which $b+W_{r} = \pm\sqrt{\beta}$ is zero,  we get from \eqref{eq_difference} and \eqref{eq_gradient} that
$\mathbb{E}_x \left[u(X_t)\right] \leq u(x)$ for all $x\in \mathbb{R}^2$ and $t\geq 0$. This completes the proof.

\section{Proof of Corollary  \ref{coro_stop}}\label{app_coro_stop}

Because \eqref{eq_opt_stop_solution} is the optimal solution to \eqref{eq_stop_prob}, by choosing $s=Y_i$, $b=W_{S_{i}+Y_i}-W_{S_i}$, and using $W_{S_{i}+Y_i+t}-W_{S_i}$ to replace $W_t$, it is immediate that \eqref{eq_opt_stop_solution1} is the optimal solution to \eqref{eq_opt_stopping}. 

The remaining task is to prove \eqref{eq_expression0} and \eqref{eq_expression}.
According to \eqref{eq_opt_stop_solution1} with $\beta\geq 0$, we have
\begin{align}
&W_{S_{i}+Y_i+Z_i}-W_{S_{i}} \nonumber\\
=&\left\{\begin{array}{l l}
W_{S_{i}+Y_i}-W_{S_{i}},&\text{if}~ |W_{S_{i}+Y_i}-W_{S_{i}}| \geq \sqrt{\beta};\\
\sqrt{\beta},&\text{if}~ |W_{S_{i}+Y_i}-W_{S_{i}}| < \sqrt{\beta}.
\end{array}\right. 
\end{align}
Hence, 
\begin{align}\label{eq_max_1}
\mathbb{E}[(W_{S_{i}+Y_i+Z_i}-W_{S_{i}})^4] = \mathbb{E}[\max(\beta^2,(W_{S_{i}+Y_i}-W_{S_{i}})^4)].
\end{align}

In addition, from \eqref{eq_zero} and \eqref{eq_expectation} we know that if  $|W_{S_{i}+Y_i}-W_{S_{i}}| \geq \sqrt{\beta}$, then \eqref{eq_zero} implies
\begin{align}
\mathbb{E}[Z_i| Y_i] = 0;
\end{align}
otherwise, if  $|W_{S_{i}+Y_i}-W_{S_{i}}| < \sqrt{\beta}$, then  \eqref{eq_expectation} implies
\begin{align}
\mathbb{E}[Z_i| Y_i] = \beta - (W_{S_{i}+Y_i}-W_{S_{i}})^2.
\end{align}
By combining these two cases, we get
\begin{align}
\mathbb{E}[Z_i| Y_i] = \max[\beta - (W_{S_{i}+Y_i}-W_{S_{i}})^2,0].
\end{align}
Using the law of iterated expectations, the strong Markov property of the Wiener process, and  Wald's identity $\mathbb{E}[(W_{S_{i}+Y_i}-W_{S_{i}})^2]=\mathbb{E}[Y_i]$, yields
\begin{align}\label{eq_max_2}
&\mathbb{E}[Z_i+ Y_i] \nonumber\\
=&\mathbb{E}[ \mathbb{E}[Z_i| Y_i] +Y_i]\nonumber\\
=&\mathbb{E}[ \max(\beta - (W_{S_{i}+Y_i}-W_{S_{i}})^2,0) +Y_i] \nonumber\\
=&\mathbb{E}[ \max(\beta - (W_{S_{i}+Y_i}-W_{S_{i}})^2,0) +(W_{S_{i}+Y_i}-W_{S_{i}})^2] \nonumber\\
=&\mathbb{E}[ \max(\beta,(W_{S_{i}+Y_i}-W_{S_{i}})^2)]. 
\end{align}
Finally, because $W_t$ and $W_{S_{i}+t}-W_{S_{i}}$ are of the same distribution,  \eqref{eq_expression0} and \eqref{eq_expression} follow from \eqref{eq_max_2} and \eqref{eq_max_1}, respectively. 
This completes the proof.

\section{Proof of Theorem  \ref{thm_zero_gap}}\label{app_zero_gap}

According to \cite[Prop. 6.2.5]{Bertsekas2003}, if we can find $\pi^{\star} = (Z_0,Z_1,\ldots)$ and $\lambda^{\star}$ satisfying
 the following conditions:
\begin{align}
&\pi^{\star}\in\Pi_1, \lim_{n\rightarrow \infty} \frac{1}{n} \sum_{i=0}^{n-1} \mathbb{E}\left[Y_i+Z_i\right] - \frac{1}{f_{\max}} \geq  0,\label{eq_mix_or_not0}\\
&\lambda^{\star}\geq 0,\label{eq_mix_or_not1}\\
&L(\pi^{\star};\lambda^{\star}) = \inf_{\pi\in\Pi_1}  L(\pi;\lambda^{\star}),\label{eq_mix_or_not}\\
&\lambda^\star \left\{\lim_{n\rightarrow \infty} \frac{1}{n} \sum_{i=0}^{n-1} \mathbb{E}\left[Y_i+Z_i\right] - \frac{1}{f_{\max}}\right\} = 0,\label{eq_KKT_last}
\end{align}
then $\pi^{\star}$ is an optimal solution to the primal problem \eqref{eq_SD} and $\lambda^\star$ is a geometric multiplier \cite{Bertsekas2003} for the primal problem \eqref{eq_SD}. 
Further, if we can find such $\pi^{\star}$ and $\lambda^{\star}$, then the duality gap between \eqref{eq_SD} and \eqref{eq_dual} must be zero, because otherwise there is no geometric multiplier \cite[Prop. 6.2.3(b)]{Bertsekas2003}. 
We note that \eqref{eq_mix_or_not0}-\eqref{eq_KKT_last} are different from the Karush-Kuhn-Tucker (KKT) conditions because of \eqref{eq_mix_or_not}.

The remaining task is to find $\pi^{\star}$ and $\lambda^{\star}$ that satisfies \eqref{eq_mix_or_not0}-\eqref{eq_KKT_last}.
According to Theorem \ref{thm_solution_form} and Corollary \ref{coro_stop}, the  solution $\pi^{\star}$ to
\eqref{eq_mix_or_not} is given by \eqref{eq_opt_stop_solution1} where $\beta = 3({\mathsf{mse}}_{\text{opt}} + \lambda^\star - \mathbb{E}\left[Y \right])$. 
In addition, as shown in the proof of Theorem \ref{thm_solution_form}, the $Z_i$'s in policy $\pi^{\star}$ are  \emph{i.i.d.}
Using \eqref{eq_mix_or_not0}, \eqref{eq_mix_or_not1}, and \eqref{eq_KKT_last}, the value of $\lambda^\star$ can be obtained  by considering two cases: 
If $\lambda^\star >0$, because the $Z_i$'s are \emph{i.i.d.}, we have 
\begin{align}\label{eq_KKT_2}
\lim_{n\rightarrow \infty} \frac{1}{n} 
\sum_{i=0}^{n-1} \mathbb{E}\left[Y_i+Z_i\right] = \mathbb{E}\left[Y_i+Z_i\right]= \frac{1}{f_{\max}}.
\end{align}
If $\lambda^\star =0$, then 
\begin{align}\label{eq_KKT_3}
\lim_{n\rightarrow \infty} \frac{1}{n}\sum_{i=0}^{n-1} \mathbb{E}\left[Y_i+Z_i\right] = \mathbb{E}\left[Y_i+Z_i\right] \geq \frac{1}{f_{\max}}.
\end{align}

Next, we use \eqref{eq_KKT_2}, \eqref{eq_KKT_3}, and $\beta = 3({\mathsf{mse}}_{\text{opt}} + \lambda^\star - \mathbb{E}\left[Y \right])$ to determine $\lambda^\star$. 
To compute ${\mathsf{mse}}_{\text{opt}}$, we substitute policy $\pi^{\star}$ and \eqref{eq_integral} into \eqref{eq_Simple}, which yields 
\begin{align}\label{eq_KKT_1}
&\mathsf{mse}_{\text{opt}}\nonumber\\
=&\lim_{n\rightarrow \infty}\!\!\frac{\sum_{i=0}^{n-1}\mathbb{E}\left[(W_{S_{i}+Y_i+Z_i}\!-\!W_{S_i})^4 \!+\! (Y_i+Z_i) \mathbb{E}[Y] \right]}{ 6\sum_{i=0}^{n-1} \mathbb{E}\left[Y_i\!+\!Z_i\right]} \nonumber\\
=&\frac{\mathbb{E}\left[(W_{S_{i}+Y_i+Z_i}\!-\!W_{S_i})^4\right]}{ 6\mathbb{E}\left[Y_i\!+\!Z_i\right]} + \mathbb{E}[Y],
\end{align}
where in the last equation we have used that the $Z_i$'s are \emph{i.i.d.} and the $(W_{S_{i}+Y_i+Z_i}-W_{S_i})$'s are \emph{i.i.d.}, which were shown  in the proof of Theorem \ref{thm_solution_form}.  
Hence, the value of $\beta = 3({\mathsf{mse}}_{\text{opt}} + \lambda^\star - \mathbb{E}\left[Y \right])$ can be obtained by considering the following two cases:

\emph{Case 1}: If $\lambda^\star >0$, then \eqref{eq_KKT_1} and \eqref{eq_KKT_2} imply that
\begin{align}\label{eq_KKT_4}
&\mathbb{E}\left[Y_i+Z_i\right]= \frac{1}{f_{\max}}, \\
&\beta > 3({\mathsf{mse}}_{\text{opt}} -\mathbb{E}[Y])  = \frac{\mathbb{E}\left[(W_{S_{i}+Y_i+Z_i}-W_{S_i})^4\right]}{ 2\mathbb{E}\left[Y_i\!+\!Z_i\right]}.
\end{align}

\emph{Case 2}: If $\lambda^\star =0$, then \eqref{eq_KKT_1} and \eqref{eq_KKT_3} imply that
\begin{align}
&\mathbb{E}\left[Y_i+Z_i\right]\geq \frac{1}{f_{\max}}, \\
&\beta = 3({\mathsf{mse}}_{\text{opt}} -\mathbb{E}[Y]) = \frac{\mathbb{E}\left[(W_{S_{i}+Y_i+Z_i}-W_{S_i})^4\right]}{ 2\mathbb{E}\left[Y_i\!+\!Z_i\right]}.\label{eq_KKT_5}
\end{align}
Combining \eqref{eq_KKT_4}-\eqref{eq_KKT_5}, yields that $\beta$ is the root of 
\begin{align}\label{eq_equation1111}
\mathbb{E}[Y_i+Z_i] \!=\! \max\left(\frac{1}{f_{\max}}, \frac{\mathbb{E}[(W_{S_{i}+Y_i+Z_i}-W_{S_i})^4]}{2\beta}\right).\!
\end{align}
Substituting \eqref{eq_expression0} and \eqref{eq_expression} into \eqref{eq_equation1111}, we obtain that $\beta$ is the root of \eqref{eq_thm1}. 
Further, \eqref{eq_opt_stop_solution1} can be rewritten as \eqref{eq_opt_solution}. Hence, if we choose $\pi^{\star}$ as the sampling policy in \eqref{eq_opt_solution} and choose $\lambda^{\star} = \beta/3 -{\mathsf{mse}}_{\text{opt}} + \mathbb{E}\left[Y \right]$ where $\beta$ is the root of  \eqref{eq_thm1}, then $\pi^{\star}$ and $\lambda^{\star}$ satisfies \eqref{eq_mix_or_not0}-\eqref{eq_KKT_last}. By using the properties of geometric multiplier mentioned above, \eqref{eq_opt_solution}  and \eqref{eq_thm1} is an optimal solution to the primal problem \eqref{eq_SD}. 

Because the problems \eqref{eq_DPExpected},  \eqref{eq_Simple}, and \eqref{eq_SD} are equivalent, \eqref{eq_opt_solution}  and \eqref{eq_thm1} is also an optimal solution to \eqref{eq_DPExpected} and \eqref{eq_Simple}. 

The optimal objective value $\mathsf{mse}_{\text{opt}}$ is given by \eqref{eq_KKT_1}. Substituting \eqref{eq_expression0} and \eqref{eq_expression} into \eqref{eq_KKT_1}, \eqref{thm_1_obj} follows. This completes the proof.

\bibliographystyle{IEEEtran}
\bibliography{ref,sueh,AgeofRealtimeMeasurements}

\begin{IEEEbiographynophoto}{Yin Sun} (S'08-M'11) is an assistant professor in the Department of Electrical and Computer Engineering at Auburn University, Alabama. He received his B.Eng. and Ph.D. degrees in Electronic Engineering from Tsinghua University, in 2006 and 2011, respectively. He was a postdoctoral scholar and research associate at the Ohio State University during 2011-2017. His research interests include wireless communications, communication networks, information freshness, information theory, and machine learning. He is the founding co-chair of the first and second Age of Information Workshops, in conjunction with the IEEE INFOCOM 2018 and 2019. The papers he co-authored received the best student paper award at IEEE/IFIP WiOpt 2013 and the best paper award at IEEE/IFIP WiOpt 2019.
\end{IEEEbiographynophoto}

\begin{IEEEbiographynophoto}{Yury Polyanskiy} (S'08-M'10-SM'14) is an Associate Professor of Electrical
Engineering and Computer Science and a member of the Laboratory for
Information and Decision Systems at the Massachusetts Institute of Technology,
Cambridge, MA, USA. He received the B.S. and M.S. degrees
in applied mathematics and physics from the Moscow Institute of Physics
and Technology, Moscow, Russia, in 2003 and 2005, respectively, and the
Ph.D. degree in electrical engineering from Princeton University, Princeton,
NJ, USA, in 2010. Currently, his research focuses on basic questions in
information theory, error-correcting codes, wireless communication, and fault-tolerant
computation. Dr. Polyanskiy won the 2013 NSF
CAREER award and the 2011 IEEE Information Theory Society Paper Award.
\end{IEEEbiographynophoto}

\begin{IEEEbiographynophoto}{Elif Uysal} (S'95-M'03-SM'13) is a Professor in the Department of Electrical and  
Electronics Engineering at the Middle East Technical University  
(METU), in Ankara, Turkey. She received the Ph.D. degree in EE from  
Stanford University in 2003, the S.M. degree in EECS from the  
Massachusetts Institute of Technology (MIT) in 1999 and the B.S.  
degree from METU in 1997. From 2003-05 she was a lecturer at MIT, and  
from 2005-06 she was an Assistant Professor at the Ohio State  
University (OSU). Since 2006, she has been with METU, and held  
visiting positions at OSU and MIT during 2014-2016. Her research  
interests are at the junction of communication and networking  
theories, with particular application to energy-efficient wireless  
networking. Dr. Uysal is a recipient of the 2014 Young Scientist Award  
from the Science Academy of Turkey, an IBM Faculty Award (2010), the  
Turkish National Science Foundation Career Award (2006), an NSF  
Foundations on Communication research grant (2006-2010), the MIT  
Vinton Hayes Fellowship, and the Stanford Graduate Fellowship. She is  
an editor for the IEEE/ACM Transactions on Networking, and served as  
associate editor for the IEEE Transactions on Wireless Communication  
(2014-2018) and track 3 co-chair for IEEE PIMRC 2019. Her current  
research is supported by TUBITAK, Huawei, and Turk Telekom.
\end{IEEEbiographynophoto}


\end{document}